\documentclass[12pt,reqno]{amsart}
\usepackage{amssymb, amscd, amsthm, mathtools}
\usepackage{url,amssymb,enumerate,colonequals}
\usepackage[margin = 1in]{geometry}
\usepackage{pdfsync}
\usepackage{multirow}
\usepackage[usenames,dvipsnames]{xcolor}
\usepackage{url}
\usepackage{microtype}
\usepackage{tikz-cd}
\usepackage[all]{xy}
\usetikzlibrary{matrix}
\usepackage{float}
\usepackage{algorithm}
\usepackage{algpseudocode}
\usepackage{tablefootnote}
\usepackage{amsmath}

\floatstyle{boxed}
\restylefloat{figure}

\newcommand{\BC}{\mathbb{C}}

\newcommand{\F}{\mathbb{F}}

\newcommand{\Z}{\mathbb{Z}}
\newcommand{\PP}{\mathbb{P}}

\newcommand{\Ga}{\alpha}
\newcommand{\Gb}{\beta}

\newcommand{\calP}{\mathcal{P}}

\newcommand{\cB}{\mathcal{B}}
\newcommand{\hcalP}{\hat{\mathcal{P}}}

\newcommand{\calQ}{\mathcal{Q}}
\newcommand{\tcB}{\tilde{\mathcal{B}}}

\newcommand{\fP}{\mathfrak{P}}

\newcommand{\cL}{\mathcal{L}}

\newcommand{\cQ}{\mathcal{Q}}

\newcommand{\GL}{\mathrm{GL}}
\newcommand{\PGL}{\mathrm{PGL}}
\newcommand{\AGL}{\mathrm{AGL}}

\DeclareMathOperator{\Div}{\mathrm{Div}}
\DeclareMathOperator{\Gal}{\mathrm{Gal}}

\DeclareMathOperator{\Aut}{\mathrm{Aut}}

\DeclareMathOperator{\Char}{\mathrm{Char}}

\DeclareMathOperator{\ddiv}{\mathrm{div}}
\DeclareMathOperator{\Supp}{\mathrm{Supp}}
\DeclareMathOperator{\ev}{\mathrm{Ev}}

\DeclareMathOperator{\bX}{\textbf{X}}
\DeclareMathOperator{\be}{\textbf{e}}
\DeclareMathOperator{\bz}{\textbf{z}}
\DeclareMathOperator{\bbe}{\underline{\textbf{e}}}

\newtheorem{theorem}{Theorem}[section]
\newtheorem{lemma}[theorem]{Lemma}
\newtheorem{corollary}[theorem]{Corollary}

\theoremstyle{definition}
\newtheorem{definition}[theorem]{Definition}

\newtheorem{example}[theorem]{Example}

\newtheorem{remark}[theorem]{Remark}

\numberwithin{equation}{subsection}

\def\enoteheading{\section*{\notesname
  \@mkboth{\MakeUppercase{\notesname}}{\MakeUppercase{\notesname}}}%
  \mbox{}\par\vskip-2.3\baselineskip\noindent\rule{.5\textwidth}{0.4pt}\par\vskip\baselineskip}
\makeatother

\begin{document}
\title[Fast Fourier Transform]{Fast Fourier transform via automorphism groups of rational function fields}

\author{Songsong Li}
\address{School of Electronic Information and Electrical Engineering, Shanghai Jiao Tong University, Shanghai 200240,  China}
\email{songsli@sjtu.edu.cn}

\author{Chaoping Xing}
\address{School of Electronic Information and Electrical Engineering, Shanghai Jiao Tong University, Shanghai 200240,  China}
\email{xingcp@sjtu.edu.cn}

\subjclass[]{}
\keywords{}
		
\maketitle	
\begin{abstract}

The Fast Fourier Transform (FFT) over a finite field $\mathbb{F}_q$ computes evaluations of a given polynomial of degree less than $n$ at a specifically chosen set of $n$ distinct evaluation points in $\mathbb{F}_q$. If $q$ or $q-1$ is a smooth number, then the divide-and-conquer approach leads to the fastest known FFT algorithms. Depending on the type of group that the set of evaluation points forms, these algorithms are classified as multiplicative (Math of Comp. 1965) and additive (FOCS 2014) FFT algorithms. In this work, we provide a unified framework for FFT algorithms that include both multiplicative and additive FFT algorithms as special cases, and beyond: our framework also works when $q+1$ is smooth, while all known results require $q$ or $q-1$ to be smooth. For the new case where $q+1$ is smooth (this new case was
not considered before in literature as far as we know), we show that if $n$ is a divisor of $q+1$ that is $B$-smooth for a real $B>0$, then our FFT needs $O(Bn\log n)$ arithmetic operations in $\mathbb{F}_q$. Our unified framework is a natural consequence of introducing the algebraic function fields into the study of FFT.

\end{abstract}

\section{Introduction}\label{sec: introduction}	
The discrete Fourier transform (DFT for short) of length $n$ over a field $F$ is a transform from $n$ coefficients of a polynomial $f(x)$ over $F$ of degree less than $n$ to $n$ evaluations of $f(x)$ at all $n$-th roots of unity in $F$. The inverse DFT (iDFT for short) is just the reverse process from $n$ evaluations to  $n$ coefficients. However, the computation of DFT directly from the definition requires $O(n^2)$ operations in $F$ which is usually too slow for practical purposes. Thus, we desire to design a faster DFT to fulfill various applications. This motivates the study of fast Fourier transform (FFT for short). FFT is a classical topic in complexity theory and has found various applications in theoretical computer science such as fast polynomial arithmetic \cite{von13, bur13, ben23ec}, and coding theory \cite{jus76,  gao03, Gao10, lin2014novel, lin16novel, lin16fft, han21}.



\subsection{Previous work}\label{subsec:pw}
The origin of FFT can be traced back to  Gauss's unpublished work in 1805, which was pointed out by Heideman et al. in \cite{heideman1984}. After 160 years, Cooley and Tukey~\cite{cooley1965} independently rediscovered this algorithm and popularized it. Now this FFT is known as the Cooley-Tukey algorithm. Cooley-Tukey's algorithm is a divide-and-conquer algorithm that implements DFT over the complex numbers $\BC$; namely, given a polynomial $f(x)=\sum_{i<n} a_ix^i \in\BC[X]$, evaluate $f(x)$ at the $n$-th roots of unity in $\BC$. If $n$ is a $B$-smooth number, i.e., all its prime factors are not greater than $B$, Cooley-Tukey's algorithm computes DFT in $O(Bn\log n)$ arithmetic operations of $\BC$. In the following, we also call $n$ is smooth if $B=O(1)$ is a constant. Using Cooley-Tukey's algorithm,  DFT over $\F_q$ can also be done in $O(n\log n)$ field operations of $\F_q$ if $\F_q$ contains an $n$-th root of unity for smooth integer $n$~\cite{Pollard71}. For DFTs of polynomials in $\F_q[x]_{<n}$, but the evaluation points in an extension field $\F_{q^d}$ which has smooth $n$-th roots of unity, van der Hoeven and Larrieu \cite{van17} showed that the Frobenius automorphism $\Phi \in\Gal(\F_{q^d}/\F_q)$ can be used to accelerate the FFT over $\F_{q^d}$. 
 The above class of FFTs is called multiplicative FFT due to the fact that the evaluation set is a multiplicative subgroup of $\F_q^*$.

However, if $\F_q$ does not have the desired $n$-th roots of unity, 
for example, $\F_q=\F_{2^r}$ and $2^r-1$ is a prime number, then the multiplicative FFT is not efficient. 
To implement FFT over $\F_q$ in this case, a new type of FFT algorithm was discovered by Zhu-Wang~\cite{WZ88} and Cantor~\cite{Cantor89} independently. To distinguish Cooley-Tukey's FFT algorithm from the one by Zhu-Wang and Cantor, Mateer, and Gao \cite{Gao10} named the latter one as the additive FFT based on the fact that the evaluation set is an additive subgroup of $\F_q$. In their work \cite{Gao10}, Mateer and Gao improved the previous additive FFT algorithm \cite{WZ88, Cantor89} and showed that their additive FFT algorithm requires $O(n\log^2 n)$ additions and $O(n\log n)$ multiplications in $\F_q=\F_{2^r}$. Particularly, in case  $r$ is a power of two, their improved FFT needs $O(n\log n)$ multiplications and only $O(n\log n\cdot\log \log n)$ additions.  Four years later, Lin et al. \cite{lin2014novel} showed that additive FFT can be run in $O(n\log n)$ operations (including additions and multiplications) over $\F_{2^r}$ by taking a novel polynomial basis of $\F_{2^r}[x]_{<n}$ (the $\F_{2^r}$-vector space consisting of all polynomials over $\F_{2^r}$ of degree less than $n$).
Moreover, their new FFT algorithm is applicable to finite field $\F_{2^r}$ with arbitrary $r$.
The same group of authors later gave a new interpretation of their algorithm and presented some applications of FFT in encoding and decoding of Reed-Solomon codes \cite{lin16fft, han21}. Note that for the additive FFT given in \cite{lin2014novel, lin16fft}, it requires that every polynomial in $\F_q[x]_{<n}$ is represented under a certain basis consisting of products of linearized polynomials instead of the standard monomial basis $\{1, x, \dots, x^{n-1}\}$.

As we have seen, both multiplicative and additive FFTs over $\F_q$ have constraints. Namely, for multiplicative FFT, one requires that $q-1$ is smooth; while for additive FFT, one requires that the characteristic $p$ of $\F_q$ is a small constant. 
Recently, Ben-Sasson et al. \cite{ben23ec} made a breakthrough in the FFT-like algorithm over an arbitrary finite field.
The new algorithm is based on elliptic curves and isogenies of smooth degree, thus it is called elliptic-curve-based FFT (ECFFT for short). 
More precisely, assume $n$ is a smooth number. Let $f(x)\in\F_q[x]_{<n}$ be a polynomial and $S'\subsetneq \F_q$ a carefully selected subset of cradinality $n$. Then the multipoint evaluation (MPE for short) of $f$ at $S'$ can be done in $O(n\log n)$ operations of $\F_q$ under the representation of $f$: $f=(f(s_1), \ldots, f(s_n))$, where $S\subsetneq\F_q$ is another subset of size $n$ and $S\cap S'=\emptyset$ (This is different from previous FFTs which take as input the coefficients of $f$ under a certain basis of $\F_q[x]_{<n}$).
They made use of isogenies between elliptic curves to construct the transform from the MPE of $f$ at $S$ to the MPE of $f$ at $S'$. During the transformation, some precomputations are required. The authors did not give the total storage for the precomputations, but at least $\Omega(n)$ is required. This is an additional overhead compared to the multiplicative/additive FFTs. Besides, a constraint for the ECFFT over $\F_q$ is that length $n$ is upper bounded by $O(\sqrt{q})$, especially for odd $q$. This restricts many applications such as encoding and decoding of $q$-ary Reed-Solomon codes where code lengths $n$ are usually proportional to $q$. 


\subsection{Sketch of the multiplicative and additive FFT techniques}
To compare our results and better understand the FFT algorithm, let us sketch the idea of Cooley-Tukey's algorithm and Lin-Chung-Han's algorithm \cite{lin2014novel}, which correspond to the multiplicative and additive cases, respectively. For both cases, let $n=2^r$ and $f(x)\in\F_q[x]$ be a polynomial of degree less than $n$. Denote the complexity of FFT with respect to the number of additions and multiplications by $A(n)$ and $M(n)$, respectively.

For Cooley-Tukey's algorithm, the evaluation set is selected to be $\mu_n\subseteq \F_q^*$ which is the set of all $n$-th roots of unity, where $n$ is equal to $2^r$ for an integer $r\ge 1$. Then the square map $s(x)=x^2$ maps $\mu_n$ onto $\mu_{n/2}$. Moreover, the polynomial $f(x)$ can be decomposed as a combination of two lower-degree polynomials meanwhile, i.e., $f(x)=f_0(x^2)+x\cdot f_1(x^2)$, where $f_0, f_1\in \F_q[x]$ are polynomials of degree less than $n/2$. Thus, the DFT of $f(x)$ at $\mu_n$ can be reduced to DFTs of $f_0$ and $f_1$ at $\mu_{n/2}$ and then a combination of these two sets of $n/2$ evaluations by using $O(n)$ additions/multiplications in $\F_q$. Then the running time $A(n)$ and $M(n)$ both satisfy the recursive formula
\begin{equation}\label{eq:recurs}
A(n)=2A(n/2)+O(n), \ M(n)=2M(n/2)+O(n).
\end{equation}
After $r$ steps of recursions, we have $A(n)=M(n)=O(n\log n)$.

Lin-Chung-Han's additive FFT algorithm inherited the idea of the FFT algorithm \cite{fid72} through polynomials modular arithmetic. Assume $\F_q=\F_{2^r}=\{\alpha_i\}_{i=0}^{2^r-1}$. Then the evaluations of $f(x)\in \F_{2^r}[x]_{<2^r}$ at $\F_{2^r}$ are exactly the $2^r$ residues
\[\big( f(x)\bmod (x-\alpha_0), f(x)\bmod (x-\alpha_1), \dots, f(x)\bmod (x-\alpha_{2^r-1}) \big).\]
Note that $\prod_{i=0}^{2^r-1}(x-\alpha_i)=x^{2^r}-x$ and $f(x)=f(x)\bmod (x^{2^r}-x)$. To efficiently get these $2^r$ residues, Lin et al. introduced linearized polynomials to decompose the modulo $(x^{2^r}-x)$ into $2^r$ modulus $\{(x-\alpha_i)\}_{i=0}^{2^r-1}$ recursively. Assume
\[\{0\}=W_0\subsetneq W_1\subsetneq\cdots \subsetneq W_{r-1}\subsetneq W_r=\F_{2^r}\]
is a subspace chain of $\F_{2^r}$ and $W_i=W_{i-1}\cup(W_{i-1}+v_{i})$ for some $v_i\in\F_{2^r}$ and each $i\in[1, r]$. In particular, we have $\dim_{\F_2}(W_i)=i$ and $\{v_1,\dots, v_r\}$ constitutes a basis of $\F_{2^r}$ as a vector space over $\F_2$. Define the linearized polynomial as $\ell_i(x)=\prod_{w\in W_i}(x-w)$. Then $\deg(\ell_i)=2^{i}$ and \[\ell_i(x+\beta)=\ell_{i-1}(x+\beta)\cdot\ell_{i-1}(x+\beta+v_{i})\ \text{for\ any}\ \beta\in\F_{2^r}.\]
For any $e\in[0,2^r-1]$, assume the $2$-adic expansion of $e$ is $e=\sum_{i=0}^{r-1}e_i 2^i$. Then
\[\deg(\ell_0^{e_0}(x)\ell_1^{e_1}(x)\cdots\ell_{r-1}^{e_{r-1}}(x)\})=e.\]
Hence $\cB=\{\ell_0^{e_0}(x)\ell_1(x)^{e_1}\cdots\ell_{r-1}^{e_{r-1}}(x)\ \mid e\in [0,2^r-1]\}$ is a basis of $\F_{2^r}[x]_{<2^r}$. Under the basis $\cB$, write $f(x)=f_0(x)+\ell_{r-1}(x)\cdot f_1(x)$, where $f_0, f_1\in \F_{2^r}[x]$ are polynomials of degree less than $2^{r-1}$. Since $x^{2^r}-x=\ell_{r-1}(x)\cdot \ell_{r-1}(x+v_{r})$, by the Chinese Reminder theorem,
\[f(x)\bmod (x^{2^r}-x)=\Big(f_0(x)\bmod \ell_{r-1}(x), \big(f_0(x)+\ell_{r-1}(v_r)f_1(x)\big)\bmod \big(\ell_{r-1}(x)+\ell_{r-1}(v_r)\big)\Big).\]
Thus, the computation of $2^r$ residues $\{f(x)\bmod (x-\alpha)\}_{\alpha\in W_r}$ of $f(x)$ can be reduced to $2^{r-1}$ residues $\{f_0\bmod (x-\alpha)\}_{\alpha\in W_{r-1}}$ of $f_0(x)$ and $2^{r-1}$ residues $\{\big(f_0+\ell_{r-1}(v_r)\cdot f_1\big)\bmod (x-\alpha)\}_{\alpha\in W_{r-1}+v_r}$ of $f_0(x)+\ell_{r-1}(v_r)f_1(x)$. Since $f_0, f_1$ both have degree less than $2^{r-1}$, the computation of $f_0+\ell_{r-1}(v_r)f_1$ needs $O(n)$ additions and $O(n)$ multiplications in $\F_{2^r}$, respectively. Thus the running time $A(n)$ and $M(n)$ also satisfy equation ~\eqref{eq:recurs} leading to $A(n)=M(n)=O(n\log n)$.

So far, we have briefly described the ideas of multiplicative FFT and additive FFT. Although there are similarities between multiplicative and additive FFTs, the methods used are different. This makes great inconvenience for the generalization of these FFT algorithms and their applications.

\subsection{Our result and comparisons}
In this work, we provide a unified framework for FFT in $\F_q$ that includes both multiplicative and additive FFT. More importantly, our framework works well when either $q-1$, $q$, or $q+1$ is smooth. In other words, we give a new FFT algorithm that includes both the multiplicative DFT and the additive DFT as two special cases. Besides, we discuss a new case that has never been considered: if $n$ is a divisor of $q+1$ that is $B$-smooth for a real number $B>0$; namely, $n$ can be factored into the product of $\prod_{i=1}^rp_i$ satisfying $p_i\leq B$ for every $1\le i\le r$, then our FFT algorithm runs in $O(B\cdot n\cdot \log n)$ operations in $\F_q$.

\begin{theorem}\label{thm:1.1} Let $B>0$ be a real. Let $\F_q[x]_{<n}$ be the space of polynomials over $\F_q$ of degree less than $n$. If $n$ is a $B$-smooth divisor of either $q-1$, $q$ or $q+1$, then one can run FFT for $f(x)\in\F_q[x]_{<n}$ at a well-chosen multipoint set of size $n$ in $O(B\cdot n\cdot\log n)$ field operations of $\F_q$.
\end{theorem}

\begin{remark}
\begin{itemize}
\item[{\rm (1)}] In our framework, we need to represent a polynomial $f\in\F_q[x]_{<n}$ under a certain basis $\cB$ of $\F_q[x]_{<n}$, then implement FFT of $f$ at a well-chosen multipoint set (the roots in $\F_q$ of a polynomial of degree $n$). As we will see in Section~\ref{sec:affine}, in the case of $n\mid q-1$, then a basis $\cB$ of  $\F_q[x]_{<n}$  constructed from our framework is actually the standard monomial basis, i.e., $\cB=\{x^i\}_{i=0}^{n-1}$; in the case of $n\mid q$, then $\cB$ is constructed by the products of a series of linearized polynomials, which is the same as in \cite{lin2014novel}. 

Moreover, if $n\mid q$ is smooth and the polynomial $f$ is given under the standard basis $\{x^i\}_{i=0}^{n-1}$, then our framework can implement DFT of $f$ in $O(n\log^2n)$ operations in $\F_q$ (see Corollary~\ref{coro: MGG} in Section~\ref{sec:affine}), which is a generalization of Mateer-Gao's additive FFT over $\F_{2^r}$ to arbitrary characteristics.

Although the results for multiplicative and additive FFTs are known, we present them under a unified framework for FFT by the theory of algebraic function fields.

\item[{\rm (2)}] If both $q$ and $q-1$ are not smooth, then neither additive FFT nor multiplicative FFT in $\F_q$ is applicable. In this case, if $q+1$ has a smooth divisor $n$, then our framework ensures that there is an efficient FFT in $\F_q$. Therefore, our work loosens the restriction of FFT on the finite field to some extent. 

For instance, let $p$ be a Sophie Germain prime, i.e., $2p+1$ is also a prime. Then both $q:=2p+1$ and $q-1=2p$ are not smooth. Thus, if $p+1$ is smooth, then $q+1$ is smooth. The other instance is that $p$ is a Sophie Germain prime and $q=2^{2p+1}-1$ is a Mersenne prime, then $q$ is not smooth. Furthermore, $q-1=2^{2p+1}-2=2\left(2^{p}-1\right)\left(2^{p}+1\right)$. Thus, with high probability $q-1$ is not smooth as $ 2^{p}-1$ is a Mersenne number. On the other hand, we have $q+1=2^{2p+1}$ is smooth.

\item[{\rm (3)}] For the new case: $q+1$ is smooth, we give a practical FFT algorithm over $\F_q$ with multipoint set $\calP\subset\F_q$ in Section~\ref{sec:q+1}. Although one could also consider multiplicative FFT in the quadratic extension field $\F_{q^2}$ in this case, the corresponding multipoint set $\calP_1$ is a subgroup of $\F_{q^2}^*$. Note that $\calP_1\cap\F_q=\{1,-1\}$. Thus, the multiplicative FFT over $\F_{q^2}$ does not imply the FFT over $\F_q$. 
\end{itemize}

\end{remark}

We compare our result with known results in Subsection~\ref{subsec:pw}, see Table~1.

\begin{table}[h!]~\label{tab:comparison}
  \begin{center}
    \caption{Comparisons of our result with known FFTs over $\F_q$ 
    with complexity $O(n\log n)$}
    \begin{tabular}{|c|c|c|c|} 
    \hline\hline
                            & \textbf{evaluation set $\calP$} & \textbf{Representation of $\F_q[x]_{<n}$} & \textbf{constraints}\tablefootnote{in all cases, $n$ need to be $O(1)$-smooth.} \\ \hline
     Multiplicative  &  $\calP\leq \F_q^*$ & $f=\sum_{i=0}^{n-1}a_ix^i$ is rep.  & $n\mid q-1$  \\
    FFT~\cite{cooley1965}   & $n$-th roots of unity &  under the standard basis & \\ \hline
      Additive  & $\calP\leq \F_q$ & $f$ is rep. under a &  $n\mid q$\\
    FFT\cite{lin2014novel}     & subspace of $\F_q$  & non-standard basis $\cB$ & \\ \hline
      Elliptic curve & $\calP\subset \F_q$ &  $f$ is rep. as MPE&  $n=O(\sqrt{q})$\\
   FFT \cite{ben23ec}& $x$-coord. of an $n$-coset\tablefootnote{a coset $C=E(\F_q)/G$, where $G\leq E(\F_q)$ is a subgroup of order $n$.}  & $f=(f(s_1),\dots,f(s_n))$& \\ \hline
   \textbf{Our} & $\calP\subset \F_q$ & $f$ is rep. under a  &  $n\mid q-1,\ q$,\\ 
    \textbf{results} & roots of an $n$-poly.\tablefootnote{a polynomial determined by the automorphism subgroup of $\Aut(\F_q(x)/\F_q)$ of order $n$.}   &  (non-)standard basis $\cB$ &  or $q+1$\\ \hline\hline
    \end{tabular}
  \end{center}
\end{table}


\subsection{Our techniques}
Let $\Aut(\F_q(x)/\F_q)$ be the automorphism group consisting of all automorphisms of $\F_q(x)$ keeping elements of $\F_q$ invariant.  Let  $G$ be a subgroup of $\Aut(\F_q(x)/\F_q)$ and let $\F_q(x)^G$ be the subfield of $\F_q(x)$ fixed by $G$. If $G$ satisfies the following conditions (see Section~\ref{sec: AFF} for precise definitions):
\begin{enumerate}
\item[\rm (i)] $G$ is an abelian group with smooth order $n=|G|$, i.e., $n=\prod_{i=1}^r p_i$, where all prime divisors $p_i$ are small constant.
\item[\rm (ii)] There is a place $\cQ$ of $\F_q(x)^G$ that is totally ramified in $\F_q(x)/\F_q(x)^G$.
\item[\rm (iii)] $\F_q(x)^G$ has a rational place $\fP$ that splits completely in $\F_q(x)$. Let $\calP$ denote the set of places of $\F_q(x)$ lying over $\fP$.
\end{enumerate}
then we can construct an FFT for any polynomial in $\F_q[x]_{<n}$ at $\calP$ with complexity $O(n\log n)$ (since every rational place $P\in\calP$ is the zero of $x-\alpha$ for $\alpha\in\F_q$ and $f(P)=f(\alpha)$ for any $f\in\F_q[x]$, we identify $\calP$ as a subset of $\F_q$). Our FFT is based on the Galois theory, in order to differentiate it from the classical FFT, we name it G-FFT.

By condition (i) and the structure of finite abelian groups, $G$ has an ascending chain of subgroups
\[\{1\}=G_0\subsetneq G_1\subsetneq \cdots G_{r-1}\subsetneq G_r=G,\]
where $|G_{i}|/|G_{i-1}|=p_{i}$ for $i=1,\cdots, r$. By the Galois theory \cite{hir08}, each group $G_i$ determines a fixed subfield $\F_q(x)^{G_i}$ of $\F_q(x)$. As any subfield of $\F_q(x)$ is again a rational function field  \cite[Proposition 3.5.9]{sti09}, we assume $\F_q(x)^{G_i}=\F_q(x_i)$ for some $x_i\in \F_q(x)$.
Thus the subgroups chain leads to a tower of fields
\[
\F_q(x)=\F_q(x_0)\supsetneq \F_q(x_1) \supsetneq \cdots \supsetneq \F_q(x_r)=\F_q(x)^G.\]
Furthermore, if the pole place $P_{\infty}$ of $x$ totally ramifies in the extension $\F_q(x)/\F_q(x)^G$, then we can choose the $x_i$ such that each $x_{i+1}$ is a polynomial of $x_i$ of degree $p_{i+1}$ for $0\leq i\leq r-1$. As a result, we have $\nu_{P_{\infty}}(x_i)=-|G_i|$ for each $0\leq i\leq r$. Based on these facts, we can construct an $\F_q$-basis of the polynomial space $\F_q[x]_{<n}$ as following
\begin{equation}\label{eq:basis}
\cB=\{x_0^{e_0}x_1^{e_1}\cdots x_{r-1}^{e_{r-1}} \mid \be=(e_0,e_1,\ldots,e_{r-1})\in\Z_{p_1}\times\Z_{p_2}\times\cdots\times\Z_{p_r}\}.
\end{equation}
Thus, any $f(x)\in\F_q[x]_{<n}$ can be represented as
\begin{equation}\label{eq:de}
\begin{split}
f(x)&=\sum_{e_0=0}^{p_1-1}\sum_{e_1=0}^{p_2-1}\ldots\sum_{e_{r-1}=0}^{p_r-1} a_{\be}\cdot x_0^{e_0}x_1^{e_1}\cdots x_{r-1}^{e_{r-1}}\\
&=f_0(x_1)+x\cdot f_1(x_1)+\dots +x^{p_1-1}\cdot f_{p_1-1}(x_1),\end{split}
\end{equation}
where the second ``=" is the result of combining all terms with the same $x^{e_0}$ for $e_0=0,1,\ldots, p_1-1$. Then $f_0(x_1),\ldots,f_{p_1-1}(x_1)\in\F_q[x_1]$ are polynomials of degree less than $n/p_1$ with respect to variable $x_1$. Thus, the evaluations of $f(x)$ at $\calP$ can be reduced to the evaluations of $f_0, f_1,\dots, f_{p_1-1}$ at $\calP$ and then a combination of these $p_1$ sets of values by using $p_1O(n)$ additions/multiplications in $\F_q$. However, for every $f_j\in \F_q(x_1)$, its evaluation at every $P\in\calP$ satisfying $f_j(P)=f\big(P\cap \F_q(x_1)\big)$. According to (iii), let $\calP_1=\calP\cap \F_q(x_1)$ be the set of places of $\F_q(x_1)$ lying over $\fP$. Then $|\calP_1|=n/p_1$. Denote the complexity of FFT with respect to the number of additions and multiplications by $A(n)$ and $M(n)$, respectively. Then $A(n)$ and $M(n)$ both satisfy the recursive formula
\begin{equation}\label{eq:recurs2}
A(n)=p_1\cdot A(n/p_1)+p_1\cdot O(n), \ M(n)=p_1\cdot M(n/p_1)+p_1\cdot O(n).
\end{equation}
Similarly, the DFTs of $f_0, \dots, f_{p_1-1}$ at $\calP_1$ can be reduced to DFTs of polynomials of lower-degree with respect to the variable $x_2$ at $\calP_2=\calP_0\cap \F_q(x_2)$. Then, after $r$ steps of recursions, we obtain $A(n)=M(n)=O(B\cdot n\log n)$, where $B=\max_{1\leq i\leq r}\{ p_i\}$.

The above idea works well for a subgroup $G$ of the affine linear group $\AGL_2(q)$ as the pole place $P_{\infty}$ of $x$ totally ramifies in $\F_q(x)/\F_q(x)^G$. However, if we choose a subgroup that is not contained in $\AGL_2(q)$, then there are no rational places that totally ramify in $\F_q(x)/\F_q(x)^G$. In this paper, we consider a cyclic subgroup $G$ of $\Aut(\F_q(x)/\F_q)$ of order $q+1$. There is a place $Q$ of degree $2$ that totally ramifies in the extension $\F_q(x)/\F_q(x)^G$, and the rational place $P_{\infty}$ is splitting completely. Assume $n=q+1=2^r$. By some elementary analysis, we first construct a basis $\cB$ of the Riemann-Roch space $\cL\left(nQ\right)=\frac{1}{Q^{n}}\F_q[x]_{\leq 2n}$ (here we identify the place $Q$ with a quadratic irreducible polynomial $Q$). Then we show that a subset $\tcB$ of $\cB$ spans the $\F_q$ vector space $\frac{u_{r,0}(x)}{Q^n}\F_q[x]_{<n}$, where $u_{r,0}(x)$ is a polynomial of degree $n$ with no roots in $\F_q$ (see Lemma~\ref{lem:subbasisq+1}). In other words, we get a basis of $\F_q[x]_{<n}$, namely, $\frac{Q^n}{u_{r,0}(x)}\tcB$. 

For a polynomial $f\in\F_q[x]$ which is represented under the basis $\frac{Q^n}{u_{r,0}(x)}\tcB$, let $\tilde{f}=\frac{u_{r,0}(x)}{Q^n}f\in\cL(nQ)$. By the same recursive reduction, we first show that the multipoint evaluation $\{\tilde{f}(\alpha)\mid \alpha\in \F_q\}$ of $\tilde{f}$ can be computed via G-FFT in time $O(n\log n)$. Next, by precomputing $\frac{Q^n(\alpha)}{u_{r,0}(\alpha)}$ for all $\alpha\in \F_q$, we get $f(\alpha)=\frac{Q^n(\alpha)}{u_{r,0}(\alpha)}\tilde(\alpha)$ for every $\alpha\in \F_q$.

In the process, we need $n$ precomputations of $\{\frac{Q^n(\alpha)}{u_{r,0}(\alpha)}\mid \alpha\in \F_q\}$. Moreover, in the recursive reduction, we also need some precomputations to get around the pole places of elements in the basis $\tcB$, which will cost $O(\log n)$ storage space. Therefore, the total storage in the G-FFT algorithm is $O(n)$.



\subsection{Organization}
The paper is organized as follows.
 In Section 2, we present some preliminaries on algebraic function fields, in particular rational function fields. In Section 3,  we construct the G-FFT via affine linear subgroups $G$ of $\Aut(\F_q(x)/\F_q)$ and instantiate $G$ by the multiplicative group $\F_q^*$ and additive group $\F_q$ as two examples. Moreover, in the case $q$ is smooth, we show that G-FFT under the standard basis of $\F_q[x]_{<n}$ can be done in $O(n\log^2 n)$. In Section 4, we construct a new basis of $\F_q[x]_{<n}$ via a non-affine cyclic subgroup of $\Aut(\F_q(x)/\F_q)$ of order $n$. Then we show that the G-FFT of smooth length $n\mid q+1$ can be done in $O(n\log n)$ and give the G-FFT algorithm for the case $q+1=2^r$. Finally, Section 5 concludes our work.

\section{Preliminary}\label{sec: Preliminary}	
In this section, we will introduce some preliminaries on algebraic function fields, especially some known results about the rational function field $\F_q(x)$ and its automorphism group. We also introduce the definition of $B$-smooth groups which are similar to the $B$-smooth integers. 

\subsection{Algebraic extensions of function fields}\label{sec: AFF} We briefly introduce the theory of rational function fields in this subsection. The reader may refer to \cite{sti09}  for details.

For a prime power $q$, let $\F_q$ denote the finite field of $q$ elements. An algebraic function field over $\F_q$ in one variable is a field extension $F\supset \F_q$ such that $F$ is an algebraic extension of $\F_q(x)$ for some transcendental element $x$ over $\F_q$. In particular, the rational function field $\F_q(x)$ is an algebraic function field of one variable.  In the following, we say $F/\F_q$ is an algebraic function field with the assumption that $\F_q$ is the full constant field of $F$, i.e., every algebraic element of $F$ over $\F_q$ belongs to $\F_q$ as well.

Now let $F$ be the rational function field $\F_q(x)$ over $\F_q$. For every irreducible polynomial $P(x)\in \F_q[x]$, we define a discrete valuation $\nu_P$ which is a map from $\F_q[x]$ to $\Z\cup\{\infty\}$ given by $\nu_P(0)=\infty$ and $\nu_P(f)=a$, where $f$ is a nonzero polynomial and $a$ is the unique nonnegative integer satisfying $P^a|f$ and $P^{a+1}\nmid f$. This map can be extended to $\F_q(x)$ by defining $\nu_P(f/g)=\nu_P(f)-\nu_P(g)$ for any two polynomials $f, g\in\F_q[x]$ with $g\neq0$.
Apart from the above finite discrete valuation $\nu_P$, we have an infinite valuation $\nu_{\infty}$ defined by $\nu_{\infty}(f/g)=\deg(g)-\deg(f)$ for any two polynomials $f, g\in\F_q[x]$ with $g\neq0$. Note that we define $\deg(0)=\infty$. The set of places of $F$ is denoted by $\PP_F$.

For each discrete valuation $\nu_P$ ($P$ is either a polynomial or $\infty$), by abuse of notion we still denote by $P$ the set $\{y\in F:\; \nu_P(y)>0\}$. Then the set $P$ is called a place of $F$.
If $P=x-\Ga$, then we denote $P$ by $P_{\Ga}$. The degree of the place $P$ is defined to be the degree of the corresponding polynomial $P(x)$. If $P$ is the infinite place $\infty$, then the degree of $\infty$ is defined to be $1$. A place of degree $1$ is called rational. In fact, there are exactly $q+1$ rational places for the rational function field $F$ over $\F_q$. 

Let $\mathbb{P}_F$ be the set of all places of $F/{\F_q}$ and let $\mathbb{P}_F^{(1)}$ be the subset of $\mathbb{P}_F$ consisting of all rational places, i.e., places of degree one. 
Assume $F'/F$ is a finite algebraic extension of degree $n$ with the constant field  $\F_q$. A place $P'\in \mathbb{P}_{F'}$ is said to be lying over $P$, written as  $P'\mid P$, if $P\subseteq P'$. The ramification index of $P'$ over $P$ is denoted by $e(P'\mid P)$.
Let $P_1, \dots, P_m$ be all the places of $F'$ lying over $P$. If $e(P_i\mid P)=[F':F]$ for a place $P_i\mid P$, then $m=1$ and $P$ is said to be totally ramified in $F'$; if $e(P_i\mid P)=1$ for every $P_i$, $i=1,\dots, m$, $P$ is said to be unramified in $F'$; furthermore, if $e(P_i\mid P)=1$ and $m=[F': F]$, then $P$ is said to split completely in $F'$. An important fact that is used in the design of our G-FFT is that for every place $P$ of $F$ and a function $f\in F$ with $\nu_P(f)\ge 0$, the value $f(P')$ is constant for all places $P'$ of $F'$ lying over $P$.

\subsection{The Riemann-Roch Space}
Let $F=\F_q(x)$ be a rational function field.
For a place $Q$ of $F$ and an integer $m$, we define the Riemann-Roch space associated with $Q$ as
\[\cL(mQ):=\{f\in F^*:\; \nu_P(f)\ge 0\;\mbox{if $P\neq Q$};\;  \nu_Q(f)  \geq -m \}\cup \{0\}.\]
Then $\cL(mQ)$ is a finite-dimensional vector space over $\F_q$. If $Q$ is the pole of $x$, then $\cL(mQ)$ is the polynomial space $\F_q[x]_{\le m}$.
The dimension $\dim_{\F_q}\cL(mQ)$ is usually denoted by $\ell(mQ)$.  Then we have
\[\ell(mQ)=m\deg(Q)+1\]
for $m\ge 0$; and $\ell(mQ)=0$, otherwise.

\subsection{Subfields of the rational function field}\label{sec:subfofr}
Let $F=\F_q(x)$ be the rational function field for some transcendental element $x\in F$ over the finite field $\F_q$.
We denote by  $\Aut(F/\F_q)$ the automorphism group of $F$ over $\F_q$, i.e.,
\begin{equation}
\Aut(F/\F_q)=\{\sigma: F\rightarrow F :\; \sigma  \mbox{ is an } \F_q\mbox{-automorphism of } F\}.
\end{equation}
It is clear that an automorphism $\sigma\in \Aut(F/\F_q)$ is uniquely determined by $\sigma(x)$.
It is well known that every automorphism $\sigma\in \Aut(F/\F_q)$ is given by
\begin{equation}\label{abcd}
\sigma(x)=\frac{ax+b}{cx+d}
\end{equation}
for some constant $a,b,c,d\in\F_q$ with $ad-bc\neq0$ (see \cite{hir08}).
Denote by  $\GL_2(q)$ the general linear group of $2\times 2$ invertible matrices over $\F_q$.
Thus, every matrix $A=\left(\begin{array}{cc}a&b\\ c&d\end{array}\right)\in \GL_2(q)$ induces an automorphism of $F$ given by \eqref{abcd}.
Two matrices of $\GL_2(q)$ induce the same automorphism of $F$ if and only if they belong to the same coset of $Z(\GL_2(q))$, where $Z(\GL_2(q))$ stands for the center $\{aI_2:\; a\in\F_q^*\}$ of $\GL_2(q)$.  This implies that $\Aut(F/\F_q)$ is isomorphic to the projective linear group $\PGL_2(q):=\GL_2(q)/Z(\GL_2(q))$. Thus, we can idenitify $\Aut(F/\F_q)$ with $\PGL_2(q)$.

Consider the subgroup of $\PGL_2(q)$
\begin{equation}\label{eq:x3}
\AGL_2(q):=\left\{\left(\begin{array}{cc}a&b\\ 0&1\end{array}\right):\; a\in\F_q^*,b\in\F_q\right\}.
\end{equation}
$\AGL_2(q)$ is called the affine linear group. Every
element $A=\left(\begin{array}{cc}a&b\\ 0&1\end{array}\right)\in \AGL_2(q)$ defines an affine automorphism
$
\sigma(x)=ax+b.$


In addition to the subgroup $\AGL_2(q)$, $\PGL_2(q)$ has a cyclic subgroup with order $q+1$. Let $P(x)=x^2+ax+b$ be a primitive polynomial over $\F_q[x]$. Consider the element $\sigma=\left(\begin{array}{cc}0&1\\ -b&-a\end{array}\right)\in \PGL_2(q)$. Then the order of $\sigma$ is $q+1$ (one can refer to \cite[Lemma V.1]{JMX19}). Thus $\langle \sigma \rangle$ induces a $q+1$-cyclic subgroup of $\Aut(\F_q(x)/\F_q)$. 

For a subgroup $G$ of $\Aut(F/\F_q)$, let $F^{G}$ be the fixed subfield by $G$, i.e.,
\[F^{G}=\{u\in F :\; \sigma(u)=u,\ \text{for\ all}\ \sigma\in G\}\]
By the Galois theory \cite{hir08},
we know that $F/{F^{G}}$ is a Galois field extension with $\Gal(F/{F^{G}})=G$. Moreover, L\"uroth's Theorem~\cite[Proposition 3.5.9]{sti09} asserts that  if $\F_q\subsetneq E\subseteq F$, then $E=\F_q(z)$, where $z\in \F_q(x)$ is a rational function of $x$ and $z$ is transcendental over $\F_q$.

\begin{lemma}[\cite{hir08}]\label{lem:simple}
Assume $\F_q(z)$ is a subfield of $\F_q(x)$ and $m=[\F_q(x): \F_q(z)]$. Then $z=\frac{f(x)}{g(x)}$ for some polynomial $f(x), g(x)\in \F_q[x]$ with $g(x)\neq 0$ satisfying $\gcd(f(x), g(x))=1$ and $m=\max\{\deg f(x), \deg g(x)\}$.
In particular, $\F_q(z)=\F_q(x)$ if and only if there exist $a, b, c, d\in \F_q$ such that $z=\frac{ax+b}{cx+d}$ and $ad-bc\neq0$.
\end{lemma}

\subsection{B-smooth groups}
Let $B>0$ be a constant integer. Recall that a positive integer $m$ is called $B$-smooth if every prime factor of $m$ is upper bounded by $B$, i.e., $m=\prod_{i=1}^r p_i$ with $p_i\leq B$ for all $1\le i\le r$. Note that we do not require these $r$ prime divisors $p_i$ to be distinct. We give the definition of $B$-smooth finite groups as follows.
\begin{definition}
A finite group $G$ is called $B$-smooth if there exists a chain of subgroups of $G$:
\begin{equation}\label{eq:grpchain}
\{1\}=G_0\subsetneq G_{1}\subsetneq \dots \subsetneq G_{r-1}\subsetneq G_r=G
\end{equation}
such that $\frac{|G_i|}{|G_{i-1}|}\leq B$ for $i=1, \dots, r$.
\end{definition}

\begin{lemma}\label{lem:smooth}
Let $G$ be a finite abelian group. If $|G|$ is $B$-smooth, then $G$ is $B$-smooth.
\end{lemma}
\begin{proof}
Assume $|G|=\prod_{i=1}^r p_i$, where $p_1, \dots, p_r$ are divisors not larger than $B$.  Since $G$ is abelian, by the structure theorem for finite abelian groups \cite{KS04}, there exists a subgroup $G_{r-1}$ of order $\prod_{i=1}^{r-1} p_i$. As $G_{r-1}$ is also abelian, there is a subgroup $G_{r-2}$ of $G_{r-1}$ of order $\prod_{i=1}^{r-2} p_i$. Continuing in this fashion until $ G_0=\{1\}$, we  get a series of finite subgroups of $G$ that satisfy the condition \eqref{eq:grpchain} and $p_{j}=\frac{|G_{j}|}{|G_{j-1}|}\leq B$.
\end{proof}

\section{Fast Fourier transform via  affine subgroups of $\Aut(\F_q(x)/\F_q)$}\label{sec:affine}
Assume the affine linear group $\AGL_2(q)$ has a smooth subgroup $G$ with order $|G|=n$. We will first construct a polynomial basis $\cB$ of $\F_q[x]_{<n}$ in terms of $G$. Under this basis $\cB$, we will then show that a Galois-group-based FFT (G-FFT for short) over $\F_q$ of length $n$ can be computed in quasi-linear time of $n$. We then instantiate our G-FFT by the multiplicative group $\F_q^*$ and the additive group $\F_q$ as two examples.


\subsection{G-FFT via  affine subgroups}\label{sec:affine1}
Assume $G\leq \AGL_2(q)$ is a $B$-smooth affine subgroup with order $n$, i.e., $|G|=n=\prod_{i=1}^rp_i$ such that $p_i\le B$ for all $1\le i\le r$. Then there is an ascending subgroup chain
\[\{1\}=G_0\subsetneq G_{1}\subsetneq \dots \subsetneq G_{r-1}\subsetneq G_r=G\]
satisfying $p_{i}=\frac{|G_{i}|}{|G_{i-1}|}\leq B$ for $i=1, \dots, r$. Let $F_i=\F_q(x)^{G_{i}}$ be the fixed subfield of $\F_q(x)$ under $G_i$; namely,
\[F_i=\{u\in \F_q(x) :\; \sigma(u)=u\ \text{for\ all}\ \sigma\in G_{i}\}.\]
By the Galois theory \cite{hir08}, $F_0/F_i$ is a Galois extension with $F_0=\F_q(x)$ and $[F_0: F_i]=|G_{i}|=\prod_{j=1}^i p_j$ for $0\leq i\leq r$.
Furthermore, we have $F_i\subsetneq F_{i-1}$ resulting a tower of fields
\[\F_q(x)=F_0\supsetneq F_{1}\supsetneq \dots \supsetneq F_{r-1}\supsetneq F_r.\]

\begin{lemma}\label{lem: new basis}
Assume $G$ is a B-smooth affine subgroup with order $n$ and $n=\prod_{i=1}^r p_i$. Let $G_i$ be a subgroup of $G$ of order $\prod_{j=1}^i p_j$ and $F_i=F^{G_i}$ be the fixed subfield under $G_i$ for $0\leq i\leq r$. Then we have
\begin{itemize}
\item[(1)] For $0\leq i< r$, $F_i=\F_q(x_i)$ and $x_i=\prod_{\sigma\in G_{i}} \sigma(x)$; moreover, each $x_{i+1}$ can be written as a polynomial of $x_i$ of degree $p_{i+1}$.
\item[(2)] The set
\[\cB=\{x_0^{e_0}x_1^{e_1}\cdots x_{r-1}^{e_{r-1}}\mid \be=(e_0,e_1,\ldots,e_{r-1})\in\Z_{p_1}\times\Z_{p_2}\times\cdots\times\Z_{p_r}\}.\]
forms a basis of the $\F_q$-vector space $\F_q[x]_{<n}$.
\end{itemize}
\end{lemma}
\begin{proof}
(1)  
Let $N_i(x)=\prod_{\sigma\in G_{i}} \sigma(x)$. Since $G$ is affine, $\sigma(x)$ is a linear polynomial of $x$ for each $\sigma\in G$ by \S~\ref{sec:subfofr}. Thus $N_i(x)\in \F_q[x]$ is a polynomial of degree $|G_i|$. Let $x_i=N_i(x)$. Then $x_i\in F_i$ and $x$ is a root of $\varphi_i(T)=N_i(T)-x_i$. These lead to
\[|G_{i}|=[F_0: F_i]\leq [F_0: \F_q(x_i)]\leq |G_{i}|.\]
Thus $F_i=\F_q(x_i)$. Moreover, by definition,
\[x_{i+1}=\prod_{\sigma\in G_{i+1}}\sigma(x)=\prod_{\sigma\in G_{i+1}/G_i}\sigma(x_i).\]
By \cite[Proposition IV.2]{JMX19}, the infinity place $P_{\infty}$ totally ramifies in $\F_q(x)/\F_q(x)^G$. Let $P_{i, \infty}=P_{\infty}\cap F_i$ be the place of $F_i$ lying below $P_{\infty}$. Then $e(P_{\infty}\mid P_{i, \infty})=|G_i|$ for $0\leq i\leq r$.
Thus
 \[\nu_{P_{i, \infty}}(x_i)=\nu_{P_{\infty}}(x_i)/e(P_{\infty}\mid P_{i, \infty})=-1;\]
namely, $P_{i, \infty}$ is a pole of $x_i$. Since $P_{i, \infty}$ is totally ramified in $F_i/F_{i+1}$, for any $\sigma\in G_{i+1}/G_i=\Gal(F_i/F_{i+1})$, $\sigma(P_{i,\infty})=P_{i, \infty}$. Thus
\[\nu_{P_{i, \infty}}\left(\sigma(x_i)\right)=\nu_{P_{i, \infty}}(x_i)=-1.\]
Therefore, each $\sigma(x_i)$ is a linear polynomial of $x_i$ and $x_{i+1}$ can be seen as a polynomial of $x_i$ of degree $|G_{i+1}/G_i|=p_{i+1}$.

(2) It is easy to see that $|\cB|=n$. For any $\bX_{\be}=x_0^{e_0}x_1^{e_1}\cdots x_{r-1}^{e_{r-1}}\in \cB$,
by computation, \[\nu_{P_{\infty}} (\bX_{\be})=-\sum_{i=0}^{r-1} e_i|G_i| > -p_r\cdot|G_{r-1}|=-n,\]
thus $\bX_{\be}\in \F_q[x]_{<n}$ and $\cB\subsetneq  \F_q[x]_{<n}$. Note that, if $\be\neq \be'=(e_0',\ldots,e_{r-1}')$, assume $k$ is the largest index in $[0, r-1]$ such that $e_k\neq e_k'$. Then
\[\nu_{P_{\infty}} (\bX_{\be})\equiv \sum_{i=0}^{k}-e_i\cdot|G_i| \bmod |G_{k+1}| \neq \nu_{P_{\infty}} (\bX_{\be'})\equiv \sum_{i=0}^{k}-e_i'\cdot|G_i| \bmod |G_{k+1}|.\]
Thus $\nu_{P_{\infty}} (\bX_{\be})\neq \nu_{P_{\infty}} (\bX_{\be'})$. As each element in $\cB$ has pairwise distinct valuation at the infinity place $P_{\infty}$, they must be linearly independent over $\F_q$. Since $\F_q[x]_{<n}$ has dimension $n$, then $\cB$ must be a basis of $\F_q[x]_{<n}$.
\end{proof}

By the above Lemma, under the basis $\cB$, we can express any $f(x)\in \F_q[x]_{<n}$ as
\[f(x)=\sum_{\be} a_{\be}\bX_{\be}=f(x, x_1,\dots, x_{r-1}),\]
and $x_i, x_{i+1},\dots, x_{r}$ are polynomials of $x_{i-1}$ for $1\leq i\leq r$.

\begin{theorem}\label{thm:main1}
 Let $F$ be the rational function field $\F_q(x)$. Let $G$ be a $B$-smooth subgroup of the affine group $\AGL_2(q)$ with $|G|=n$.  Assume $\fP$ is a rational place of $F^G$ that splits completely in the extension $F/F^G$. Let $\calP$ be the set of rational places of $F$ lying over $\fP$ (hence $|\calP|=n$). Then, under the basis $\cB$ as in Lemma~\ref{lem: new basis}, the evaluations of every function $f(x)\in\F_q[x]_{<n}$ at $\calP$ can be computed in $O(B\cdot n\log n)$ operations in $\F_q$.
\end{theorem}
\begin{proof}
First, denote the complexity (i.e., the total number of operations in $\F_q$) of evaluating $f(x)$ at the set $\calP$ by $C(n)$.
Under the basis $\cB$ defined in Lemma~\ref{lem: new basis}, assume $f=f(x, x_1, \dots, x_{r-1})$, where $\deg_{x_i}(f)<p_{i+1}$. To perform FFT of $f(x)$ at $\calP$, write $f=f(x, x_1, \dots, x_{r-1})$  in at most $O(n)$ steps as follows
\begin{equation}\label{eq:rec}
f=f_0(x_1, \dots, x_{r-1})+x\cdot f_1(x_1, \dots, x_{r-1})+\dots +x^{p_1-1}\cdot f_{p_1-1}(x_1, \dots, x_{r-1}).\end{equation}
Note that
\[\begin{split}\nu_{P_{1,\infty}} f_i(x_1, \dots,x_{r-1})&\geq-\big((p_2-1)+(p_3-1)p_2+\cdots+(p_r-1)p_2\cdots p_{r-2}\big)\\
&>-p_2\cdots p_r=-n/p_1\ \text{for}\ i=0,\cdots, p_1-1,\end{split}\]
where $P_{1, \infty}=P_{\infty}\cap F_1$ is the infinity place of $F_1$. Thus we have $f_0, \dots, f_{p_1-1}\in \cL\left((n/p_1)P_{1, \infty}\right)=\F_q[x_1]_{<n/p_1}$. By equation \eqref{eq:rec}, the DFT of $f$ at $\calP$ can be reduced to the DFTs of $f_0, f_1,\dots, f_{p_1-1}$ at $\calP_{1}=\calP\cap F_1$ and then a combination of these values by using $p_1O(n)$ operations in $\F_q$. Therefore, the running time $C(n)$ satisfies the following recursive formula
\begin{equation}\label{eq:time}
C(n)=p_1C(n/{p_1})+p_1\cdot O(n).
\end{equation}
We continue the reduction in this fashion till to the DFTs of functions in $\F_q[x_{r}]_{<1}$. We get a total of $(p_1\cdots p_{r})$ constant polynomials and need to compute their's evaluations at $\calP_{r}=\calP\cap F_{r}=\{\fP\}$, which can be done in $(p_1\cdots p_r)O(1)$ orperations. Thus, the recursive formula \eqref{eq:time} leads the total running time equal to
\[C(n)=p_1\cdots p_r\cdot O(1)+(p_1+\dots+p_{r})O(n)=O(B\cdot n\log n).\]
This completes the proof.\end{proof}

\begin{remark}\label{rmk:ifft}
Let $\ev_{\calP}(f)$ denote the MPE of $f\in \F_q[x]_{<n}$ at $\calP$. The inverse G-FFT means that given the $\ev_{\calP}(f)$, output the coefficients of $f(x)$ under the basis defined by G-FFT (see Lemma~\ref{lem: new basis}). In the following, we show that the inverse  G-FFT can also be done in $O(n\log n)$ operations. Let $I(n)$ denote the complexity of inverse G-FFT of length $n$. For any $P^{(1)}\in\calP\cap F_1=\calP^{(1)}$, let $P_{1,1},\ldots,P_{1,p_1}$ be the $p_1$ places in $\PP_{F}$ lying over $P^{(1)}$. By the recursive Equation~\eqref{eq:rec}, 
\begin{equation}\label{eq:irec}
\left(\begin{matrix} f_0(P^{(1)})\\
f_1(P^{(1)})\\
\vdots\\
f_{p_1-1}(P^{(1)})\end{matrix}\right)=\left(\begin{matrix}
1            &            x(P_{1,1})    & \cdots  & x^{p_1-1}(P_{1,1})\\ 
1 & x(P_{1,2}) & \cdots  & x^{p_1-1}(P_{1,2}) \\
\vdots       & \vdots       & \cdots   & \vdots \\
1 &   x(P_{1,p_{1}}) & \cdots &  x^{p_1-1}(P_{1,p_{1}})\end{matrix}\right)^{-1}\cdot \left(\begin{matrix} f(P_{1,1})\\
f(P_{1,2})\\
\vdots\\
f(P_{1,p_1})\end{matrix}\right)\end{equation}
Thus, we can first compute the MPEs of $f_0, f_1,\ldots, f_{p_1-1}$ at $\calP^{(1)}$ from the above equation, which will cost $p_1O(n)$ operations. Then we can reduce the inverse G-FFT of $\ev_{\calP}(f)$ to $p_1$ inverse G-FFTs of $\ev_{\calP^{(1)}}(f_0)$, $\ev_{\calP^{(1)}}(f_1)$, $\ldots$, $\ev_{\calP^{(1)}}(f_{p_1-1})$. Finally, we get $f$ by Equation~\eqref{eq:rec}. Therefore, $I(n)$ satisfies the following recursive formula
\[I(n)=p_1I(n/{p_1})+p_1O(n)=O(n\log n).\]
\end{remark}

\subsection{Instantiation}
\subsubsection{Multiplicative G-FFT} Let us first look at the multiplicative case. Let $F$ be the rational function field $\F_q(x)$. Consider a subgroup $T$ of $\F_q^*$ of order $n$ and an automorphism subgroup
\[G_T=\{\sigma\in\Aut(F/\F_q):\; \sigma(x)=ax\;\mbox{for some $a\in T$}\}.\]
When $T=\F_q^*$, we denote $G_T$ by $G_*$. Then the field extension $F/F^{G_*}$ has extension degree $[F:F^{G_*}]=q-1$. Furthermore, there are a few facts about this extension (refer to \cite[Proposition IV.2]{JMX19})
\begin{enumerate}
\item The pole of $x$ is totally ramified in the extension $F/F^{G_*}$.
\item There is a rational place $\wp$ of $F^{G_*}$ that splits completely in $F/F^{G_*}$.
\end{enumerate}
Let $\fP$  be a rational place of $F^{G_T}$ that lies over $\wp$. Then $\fP$ splits completely in the extension $F/F^{G_T}$. Let $\calP$ be the set of rational places of $F$ lying over $\fP$. Then $|\calP|=n$.

Assume that $n$ has factorization $n=\prod_{i=1}^rp_i$ with $p_i\le B$ for a positive integer $B$. Consider an ascending chain of subgroups
\[\begin{split}
\{1\}=T_0\subsetneq T_{1}\subsetneq \dots \subsetneq T_{r-1}\subsetneq T_r=T. \end{split}\]
with $|T_i|=\prod_{j=1}^ip_i$ for $i\ge 1$. Let $G_i=G_{T_i}$ and $F_i=F^{G_i}$. If $\Char(\F_q)=2$, let $x_i:=\prod_{\sigma\in G_i} \sigma(x)$. Then $F_i=\F_q(x_i)$ and $x_i=(\prod_{a\in T_i} a)\cdot x^{\prod_{j=1}^i p_j}=x^{\prod_{j=1}^i p_j}$. If $\Char(\F_q)$ is odd, then $2\mid q-1$. We choose $T_1$ with even order (hence $n$ is also even). Then $2\mid |T_i|$ and $\prod_{a\in T_i} a=-1$ for each $T_i$.  Let $x_i=-\prod_{\sigma\in G_i}\sigma(x)$. Then $F_i=\F_q(x_i)$ and $x_i=x^{\prod_{j=1}^i p_j}$. Therefore, we always have $x_{i}=x_{i-1}^{p_i}$ for all $i\ge 1$. Furthermore, it is easy to see that $x_i=0$ gives only one solution $x=0$, i.e., the zero of $x_i$ totally ramifies in $F/F_i$. Let $\Ga$ be a primitive element of $\F_q$. Then for any $\Gb\in\F_q^*$, the equation $x_r=\Gb^n$, i.e., $x^{n}=\Gb^n$ gives solutions $x=\Gb\Ga^{j(q-1)/n}$ for $j=0,1,\dots,n-1$. This implies that $x_r-\Gb^n$ splits into $n$ rational places of $F$. Let $\calP$ be the set of rational places of $F$ lying over  $x_r-\Gb^n$, i.e., $\calP=\{x-\Gb\Ga^{j(q-1)/n}:\; j=0,1,\dots,n-1\}$. Thus, we have $\calP_i=\calP\cap F_i=\{x_i-\Gb^{\prod_{j=1}^ip_i}\Ga^{k(q-1)\prod_{j=1}^ip_j/n}:\; k=0,1,\dots, \frac{n}{\prod_{j=1}^ip_j}-1\}$.  Partition $\calP$ into $\frac{n}{\prod_{j=1}^ip_j}$ subsets $\Re_k:=\{(x-\Gb\cdot\Ga^{k(q-1)/n}\cdot \Ga^{\ell(q-1)/\prod_{j=1}^ip_j}):\; \ell=0,1,\dots,\prod_{j=1}^ip_j-1\}$ for $k=0,1,\dots, n/(\prod_{j=1}^ip_j)-1$. It is easy to see that evaluation of $x_i$ at every place of $\Re_k$ is a constant that is equal to $\Gb^{\prod_{j=1}^ip_i}\Ga^{k(q-1)\prod_{j=1}^ip_j/n}$.

In our G-FFT algorithm, we have to decompose a polynomial in variable $x_{i-1}$ in terms of polynomials in variable $x_i$. So we need represent $f\in\F_q[x]_{<n}$ under the basis $\cB$ defined in Lemma~\ref{lem: new basis}. Actually, in this case,
\[\begin{split}
\cB_T&=\{x_0^{e_0}x_1^{e_1}\cdots x_{r-1}^{e_{r-1}}\mid \be=(e_0,e_1,\ldots,e_{r-1})\in\Z_{p_1}\times\Z_{p_2}\times\cdots\times\Z_{p_r}\}\\
&=\{x^e\mid e=\sum_{i=0}^{r-1}e_i\cdot(p_1\cdots p_{i+1})\in[0,n-1]\}=\{1,x,\cdots,x^{n-1}\},
\end{split}\]
namely, $\cB_T$ is exactly the standard basis of $\F_q[x]_{<n}$. To illustrate, let us consider the first step, i.e., express a polynomial $f(x)\in\F_q[x]_{<n}$ as a combination of $\{x^k\}_{k=1}^{p_1-1}$ with coefficients of polynomials in variable $x_1=x^{p_1}$. Then, in at most $O(n)$ steps, $f(x)$ can be written as
\begin{equation}
\begin{split}f(x)&=f_0(x^{p_1})+x f_1(x^{p_1})+\cdots+x^{p_1-1}f_{p_1-1}(x^{p_1})\\
&=f_0(x_1)+x f_1(x_1)+\cdots+x^{p_1-1}f_{p_1-1}(x_1).\end{split}
\end{equation}
It is clear that each of $f_i(x_1)$ has degree less than $n/p_1$. We can then decompose each $f_i(x_1)$ as a combination of $\{x_1^k\}_{k=0}^{p_2-1}$ with coefficients of polynomials in variable $x_2=x_1^{p_2}$. Then the polynomials in $x_2$ all have degrees less than $n/(p_1p_2)$. We continue in this fashion until in the last step, all polynomials in variable $x_r$  have degrees less than $n/\prod_{j=1}^rp_j=1$, i.e., they are all constants.

In conclusion, we have the following result.
\begin{corollary} If $n$ is a divisor of $q-1$ that is $B$ smooth, then one can run FFT for polynomials $f(x)\in\F_q[x]_{<n}$ in $O(B\cdot n\cdot\log n)$ field operations of $\F_q$.
\end{corollary}

\subsubsection{Additive G-FFT}
We now turn to the additive case. Let $F$ be the rational function field $\F_q(x)$. Let $p$ be the characteristic of $\F_q$. Consider an  $\F_p$-subspace $W$ of $\F_q$ of order $n$ and the associated automorphism group
\[G_W=\{\sigma\in\Aut(F/\F_q):\; \sigma(x)=x+b\;\mbox{for some $b\in W$}\}.\]
When $W=\F_q$, we denote $G_W$ by $G_+$. Then the field extension $F/F^{G_+}$ has extension degree $[F: F^{G_+}]=q$. Furthermore, there are a few facts about this extension \cite{JMX19}:
\begin{enumerate}
\item The pole  of $x$ is totally ramified in the extension $F/F^{G_+}$.
\item There is a rational place $\wp$ of $F^{G_+}$ that splits completely in $F/F^{G_+}$.
\end{enumerate}
Let $\fP$  be a rational place of $F^{G_W}$ that lies over $\wp$. Then $\fP$ splits completely in the extension $F/F^{G_W}$. Let $\calP$ be the set of rational places of $F$ lying over $\fP$. Then $|\calP|=n$.

Assume that $n$ has factorization $n=p^r$. For simplicity, choose a set $\{\Ga_1,\Ga_2,\dots,\Ga_r\}$ of $\F_p$-linearly independent elements in $\F_q$. Put $W_i=\sum_{j=1}^i\F_p\Ga_j$. Then $W_i$ is an $\F_p$-subspace of $\F_q$ of dimension $i$. Hence, $|W_i|=p^i$ for $1\le i\le r$. Put $G_i=\{\sigma\in\Aut(F/\F_q):\; \sigma(x)=x+b\;\mbox{for some $b\in W_i$}\}$, $F_i=F^{G_i}$ and $x_i=\prod_{\sigma\in G_i}\sigma(x)=\prod_{a\in W_i}(x-a)$. Then each $x_i$ is a linearized polynomial of degree $p^i$, denoted by $\ell_i(x)$. Furthermore, $x_i=x_{i-1}^p-\Gb_i x_{i-1}$ where $\Gb_i=\ell_{i-1}^{p-1}(\alpha_i)\in\F_q$. It is easy to see that $x_i=0$ gives solutions $x=a$ for all $a\in W_i$. This implies that $x_r$ splits into $n$ rational places of $F$. Let $\calP$ be the set of rational places of $F$ lying over  $x_r$, i.e., $\calP=\{x-a:\; a\in W\}$. Thus, there exists an $\F_p$-vector space $V_i$ of dimension $r-i$ such that $\calP_i=\calP\cap F_i=\{x_i-a:\; a\in V_i\}$.

To perform G-FFT, we need to represent $f(x)\in \F_q[x]_{<n}$ under a basis $\cB=\cB_W$ constructed in Lemma~\ref{lem: new basis}. In this case, each $x_i=\ell_i(x)=x_{i-1}^p-\Gb_i x_{i-1}$ is a linearized polynomial of degree $p^i$ for $0\leq i\leq r$. In particular, $\ell_0(x)=x$. Thus,
\begin{equation}\label{eq:basisq}
\cB_W=\{\ell_0(x)^{e_0}\ell_1(x)^{e_1}\cdots\ell_{r-1}(x)^{e_{r-1}}:\ (e_0,\dots, e_{r-1})\in \Z_p^r\}.
\end{equation}
We note that $\cB_W$ is constructed in the same way as in \cite{lin2014novel}. Under this basis, we can decompose a polynomial in variable $x_{i-1}$ in terms of polynomials in variable $x_i$ step by step. To illustrate, let us consider the first step only. Let $f(x)$ be a polynomial of degree less than $n$. Under the basis $\cB_W$, write $f(x)$ as
\begin{equation}\label{eq:q}
\begin{split} f(x)&=f_0(x_1)+x f_1(x_1)+\cdots+x^{p_1-1}f_{p_1-1}(x_1)\\
&=f_0(x^{p}-\Gb_1x)+x f_1(x^{p}-\Gb_1x)+\cdots+x^{p-1}f_{p_1-1}(x^{p}-\Gb_1x).\end{split}
\end{equation}
It is clear that each of $f_i(x_1)$ has degree less than $n/p=p^{r-1}$. We can then decompose each $f_i(x_1)$ as a combination of $\{x_1^i\}_{i=1}^{p-1}$ and polynomials in variable $x_2=x_1^{p}-\Gb_2x_1$. Then the polynomials in $x_2$ all have degrees less than $n/p^2=p^{r-2}$. We continue in this fashion until the last step: all polynomials in variable $x_r$ have degrees less than $n/p^r=1$, i.e., all constant polynomials. In conclusion, we have the following result.
\begin{corollary}\label{coro:q}
Let $\cB_W$ be a basis of $\F_q[x]_{<n}$ defined as in equation~\eqref{eq:basisq}. Let $f(x)\in\F_q[x]_{<n}$ be a given polynomial that is represented under $\cB_W$. If $n$ is a divisor of $q$, then one can run FFT for $f(x)$ in $O(p\cdot n\cdot\log n)$ field operations of $\F_q$. \end{corollary}

The condition ``$f(x)\in\F_q[x]_{<n}$ is represented under $\cB_W$" in Corollary~\ref{coro:q} is necessary to decompose a polynomial in variable $x_{i-1}$ in terms of polynomials in variable $x_i$ step by step in our G-FFT. If we are given a representation of $f(x)$ under the standard basis, i.e., $f(x)=\sum_{i=0}^{n-1}a_ix^i$, we need to compute its $(x^p-\Gb_1 x)$-adic expansion (refer to \cite{Gao10, von13}) first, i.e.,
\[f(x)=a_0(x)+a_1(x)(x^p-\Gb_1 x)+\dots+a_m(x)(x^p-\Gb_1x)^m,\ \deg(a_i)<p\ \text{for}\ 0\leq i\leq m<p^{r-1}.\]
Then we can write $f(x)$ as in the form of \eqref{eq:q} in at most $O(n)$ steps from its $(x^p-\Gb_1 x)$-adic expansion. In the second recursion, we need also to compute the $(x_1^p-\Gb_2 x_1)$-adic representation of each $f_i(x_1)$ to write it as a combination of $\{x_1^i\}_{i=1}^{p-1}$ and polynomials in variable $x_2=x_1^{p}-\Gb_2x_1$. We continue in this fashion until the last step: all polynomials in variable $x_r$ are constant. This has been considered in \cite{Gao10} for the case that $\F_q=\F_{2^r}$. We notice that if the characteristic $p$ is a constant (not only for $p=2$), then the $(x^p-\beta x)$-adic expansion of any $f(x)\in\F_q[x]_{<n}$ can be computed in $O(n\log n)$ for any $\beta\in \F_q$ (see Lemma~\ref{lem:poly-adic}). Consequently, if $f(x)=\sum_{i=0}^{n-1}a_i x^i$ is given under the standard basis, we must add the complexity of computing the $(x^p-\beta x)$-adic expansions in the recurse formula of $C(n)$. Then $C(n)$ satisfies the following recursive formula
\[C(n)=pC(n/p)+pO(n)+O(n\log n),\]
After $r$ steps of recursion, we have
\[C(n)=p^rC(1)+O\left(n\cdot(\log p+\cdots+\log n)\right)+r\cdot p\cdot O(n)=O(n\log^2 n)\]
 In conclusion, we have the following result.

\begin{corollary}\label{coro: MGG}
Let $\F_q$ be a finite field with constant characteristic $p$. Let $f(x)\in\F_q[x]_{<n}$ be a given polynomial under the standard basis.  If $n$ is a divisor of $q$, then one can run FFT for $f$ in $O(n\log^2 n)$ field operations of $\F_q$.
\end{corollary}

\section{Fast Fourier transform via  a cyclic subgroup of order $q+1$}\label{sec:q+1}
We continue to use the notation $F$ to denote the rational function field $\F_q(x)$. Assume $n\mid q+1$ is a $O(1)$-smooth divisor. Let $m(x)=x^2+ax+b\in\F_q[x]$ be an irreducible and primitive polynomial over $\F_q$ and $Q(x)=x^2+\frac a{b}x+\frac 1{b}$. Let $Q$ denote the quadratic place of $F$ corresponding to $Q(x)$. In this section, we will show that the DFT of any polynomial $f(x)\in\F_q[x]_{< n}$ at a well-chosen set can be done in $O(n\log n)$ by using our G-FFT algorithm.

Recall that $\sigma=\left(\begin{matrix}0&1\\ -b&-a\end{matrix}\right)$ has order $q+1$ in the group $\Aut(F/\F_q)=\PGL_2(q)$, where $a, b$ are coefficients of $m(x)$. Let $\PP_F^{(1)}$ be the set of $q+1$ rational places of $F$. Then $\sigma$ acts as a $(q+1)$-cycle on the set $\PP_F^{(1)}$.

\begin{lemma}[\cite{JMX19}]\label{lem:q+1subf}
Let $Q(x)$ and $\sigma$ be defined as above. Let $G\leq \langle \sigma\rangle$ be a subgroup of order $n$. Then the fixed subfield of $F$ by $G$ is $F^G=\F_q(z),\ \text{where}\ z=\sum_{\tau\in G}\tau(x).$ Moreover,
\begin{itemize}
\item[(1)] the pole place $\fP\in\PP_{F^G}$ of $z$ splits completely in $F/F^G$;
\item[(2)] the quadratic place $Q\in\mathbb{P}_{F}$ is the unique place that is totally ramified in $F/F^G$.
\end{itemize}
\end{lemma}

Assume $G=\langle \sigma^{(q+1)/n}\rangle$ is a $B$-smooth subgroup of order $n$, i.e., $|G|=n=\prod_{i=1}^r p_r$, where $p_1,\dots,p_r\leq B$. Then there is an ascending chain of subgroups:
\[G_0=\{1\},\ G_{i-1}\leq G_{i}\ \text{with\ index}\ [G_i: G_{i-1}]=p_i,\ i=1, \dots, r.\]
Let $F_i=F^{G_{i}}$ be the fixed subfield. By the same analysis as in Subsection~\ref{sec:affine1}, we have $F_i\subsetneq F_{i-1}$ resulting a tower of fields
\[F_0=\F_q(x),\ F_{i-1}\supsetneq F_{i}\ \text{with\ extension\ degree}\ [F_{i-1}: F_{i}]=p_i,\ i=1, \dots, r.\]
Define
\begin{equation}\label{eq:defxy}
x_i=\sum_{\tau\in G_i} \tau(x),\ y_i=\prod_{\tau\in G_i} \tau\left(\frac{1}{Q(x)}\right).
\end{equation}
Then $F_i=\F_q(x_i)$ by taking $G=G_i$ in Lemma~\ref{lem:q+1subf}. Let $E_i=\F_q(y_i)$. It is clear that $E_i\subsetneq F_i$ according to the definition of $F_i$. Furthermore, since $Q$ is totally ramified in $F/F^G$, $\tau(Q)=Q$ for all $\tau\in G$ \cite[Theorem 3.8.2]{sti09}. Thus the pole divisor of $y_i$ is $(y_i)_{\infty}=\sum_{\tau\in G_i}\tau(Q)=|G_i|Q$. By \cite[Theroem 1.4.11]{sti09},
\[ [F: E_i]=\deg (y_i)_{\infty}=2|G_i|.\]
Thus $[F_i: E_i]=[F:E_i]/[F:F_i]=2$. We have the following diagram
\[\begin{tikzcd}[row sep=tiny, column sep = large]
	\F_q(x) \ar[dd, dash,"p_1"]\ar[dddr, dash, "2p_1"]  &        \\
	                              &                 \\
	\F_q(x_1)\ar[dd, dash,"p_2"]\ar[rd, dash, "2"]\ar[dddr, dash, "2p_2"]         &          \\
	                              &         \F_q(y_1)    \\
	\F_q(x_2)\ar[dd, dash, "p_3"]      \ar[rd, dash, "2"]  &           \\
	                              &          \F_q(y_2)     \\
	 \vdots\ar[dd, dash, "p_r"]      \ar[dddr, dash, "2p_r"]                                  &           \\
	                              &         \vdots           \\
	  \F_q(x_r) \ar[rd, dash, "2"]                                     &           \\
	                              &  \F_q(y_r).                                             	
\end{tikzcd}\]

Let $P_{i,\infty}\in\PP_{F_i}$ denote the pole place of $x_i$ for $0\leq i\leq r$. According to Lemma~\ref{lem:q+1subf}, we know that $P_{i,\infty}$ splits completely in $F/F_i$. In particular, let $\calP\subset\PP_F$ be the set of places lying over $P_{r,\infty}$. Then $|\calP|=n$. Let $\calP^{(i)}=\calP\cap F_i$ be the set of all rational places in $\PP_{F_i}$ lying above $P_{r,\infty}$. Then $|\calP^{(i)}|=n/{|G_i|}$. By the definition of $x_i$, $P_{0,\infty}\mid P_{i,\infty}$ for every $i$, thus $P_{i,\infty}\in \calP^{(i)}$. Define
\begin{equation}\label{eq:calQ_i}
\hcalP^{(i)}=\calP^{(i)}\setminus \{P_{i,\infty}\},\ i=0,1,\ldots,r-1.
\end{equation}
Since every rational place $P\in \hcalP^{(i)}$ is the zero of $x_i-\alpha$ for some $\alpha\in \F_q$, we identify $P$ with $\alpha$ in this correspondence. Thus, $\hcalP^{(i)}$ can be viewed as a subset of $\F_q$ for all $i$. In particular, $\calP^{(0)}=\calP$ and $\hcalP=\calP\setminus\{P_{\infty}\}$.

The following lemma presents the relationships between $x_{i}, y_{i}, x_{i-1}, y_{i-1}$, which is crucial for us to construct a basis of $\F_q[x]_{<n}$ and perform the G-FFT for any $f\in\F_q[x]_{<n}$.

\begin{lemma}\label{lem:key lemma}
For $1\leq i\leq r$, let $x_i$, $y_i$ be defined as in Equation~\eqref{eq:defxy}. Let $Q_i=Q\cap F_i$ be the quadratic place of $F_i$ lying below $Q$. Then, the followings hold:
\begin{itemize}
\item[(1)] $1/{y_i}$ is a prime element of $Q_i$, i.e., $1/y_i={Q_i(x_i)}$, where $Q_i(x_i)$ is a quadratic irreducible polynomial corresponding to $Q_i$.
\item[(2)] The pole place $P_{i,\infty}$ of $x_i$ splits into $p_i$ rational places including $P_{i-1,\infty}$ in $F_{i-1}$; namely, $x_i=\frac{u_{i, i-1}(x_{i-1})}{(x_{i-1}-\lambda_{i,1})\cdots (x_{i-1}-\lambda_{i,|p_i|-1})}$ for some polynomial $u_{i, i-1}(T) \in \F_q[T]$ of degree $p_i$, and $\lambda_{i,1}, \dots, \lambda_{i,|p_i|-1}\in \F_q$ are pairwise distinct.
\item[(3)] $y_i=c_iy_{i-1}^{p_i}\cdot(x_{i-1}-\lambda_{i,1})^2\cdots (x_{i-1}-\lambda_{i,|p_i|-1})^2$, where $c_i\in\F_q^*$ and $\lambda_{i,1}, \dots, \lambda_{i,|p_i|-1}$ are given in (2).
\item[(4)] $x_r=\frac{u_{r,i}(x_i)}{\prod_{\alpha\in\hcalP^{(i)}}(x_i-\alpha)}$ for some $u_{r,i}(T)\in\F_q[T]$ of degree $n/|G_i|$, and $y_r=\frac{\prod_{\alpha\in\hcalP^{(i)}}(x_i-\alpha)^2}{c_{r,i}Q_i^{n/|G_i|}(x_i)}$, where $c_{r,i}\in\F_q^*$ and $\hcalP^{(i)}$ is defined by Equation~\eqref{eq:calQ_i}.
\end{itemize}
\end{lemma}
\begin{proof}
Let $\tau_i=\sigma^{(q+1)/|G_i|}$. Then $\tau_i$ is a generator of $G_i$, i.e., $G_i=\langle \tau_i\rangle$. Since $\sigma$ acts as a $(q+1)$-cycle on the set $\mathbb{P}_{\F_q(x)}^{(1)}$, then, as a permutation of $\mathbb{P}_{\F_q(x)}^{(1)}$, $\tau_i$ can be decomposed as a product of $(q+1)/|G_i|$ cycles of length $|G_i|$. Assume the cycle which contains $P_{0,\infty}=P_{\infty}$ is as follows
\begin{equation}\label{eq:cycle}
P_{\infty}\stackrel{\tau_i}{\longrightarrow}P_{\alpha_{i, 1}}\stackrel{\tau_i}{\longrightarrow}\dots \stackrel{\tau_i}{\longrightarrow}P_{\alpha_{i, {|G_i|-1}}}\stackrel{\tau_i}{\longrightarrow}P_{\infty},
\end{equation}
where $\alpha_{i,1},\dots,\alpha_{i,|G_i|-1}\in \F_q$ are pairwise distinct. Thus
\begin{equation}\label{eq:usfx}
x_i=x+\tau_i(x)+\cdots \tau_i^{|G_i|-1}(x)=\frac{u_i(x)}{(x-\alpha_{i,1})\cdots (x-\alpha_{i,|G_i|-1})},
\end{equation}
for some $u_i(x) \in \F_q[x]$. Consider the principal divisor $\ddiv(1/Q(x))=2P_{\infty}-Q$. Since $Q$ is totally ramified, $\tau_i^{j}(Q)=Q$ for all $\tau_i^j\in G_i$ by \cite[Theorem 3.8.2]{sti09}. Then \[\nu_{Q}\left( \tau_i^{j}(1/Q(x))\right)=\nu_{\tau_i^{-j}(Q)}(1/Q(x))=\nu_{Q}(1/Q(x))=-1.\]
 According to equation~\eqref{eq:cycle},
\[\ddiv(y_i)=\sum_{j=0}^{|G_i|-1}\ddiv(\left( \tau_i^{j}(1/Q(x))\right)=2P_{\infty}+\sum_{j=1}^{|G_i|-1}2P_{\alpha_{i, j}}-|G_i|Q.\]
Thus, we have
\begin{equation}\label{eq:usfy}
y_i=c_i'\cdot\frac{(x-\alpha_{i,1})^2\cdots (x-\alpha_{i,|G_i|-1})^2}{Q^{p_1\cdots p_i}(x)},\ \text{for\ some} \ c_i'\in\F_q^*.
\end{equation}
Let $P_{i, \infty}\in \mathbb{P}_{F_i}$ be the pole of $x_i$ for $0\leq i\leq r$. By Equations \eqref{eq:usfx} and \eqref{eq:usfy}, $P_{i, \infty}$ is a double zero of $y_i$ and $Q_i$ is the unique pole of $y_i$ in $F_i$. Furthermore,
\[\nu_{Q_i}(y_i)=\nu_Q(y_i)/e(Q\mid Q_i)=-1.\]
Thus $1/y_i$ must be a prime element of $Q_i$, denoted by $Q_i(x_i)$. Then $y_i=1/Q_i(x_i)$. This completes the proof of (1).

(2) For the relationship between $x_i$ and $x_{i-1}$, we note that
\begin{equation}\label{eq:i&i-1}
x_i=\sum_{\bar{\tau}\in G_{i}/G_{i-1}} \bar{\tau}(x_{i-1}).
\end{equation}
By Equation~\eqref{eq:usfx}, the pole $P_{i, \infty}$ of $x_i$ splits completely in the extension $F/F_i$. Thus $P_{i, \infty}$ splits into $p_i$ places of $F_{i-1}$. It is obvious to see that $P_{i-1,\infty}\mid P_{i,\infty}$ from equation~\eqref{eq:i&i-1}. Assume $P_{i-1, \infty}, P_{\lambda_{i, 1}},\dots, P_{\lambda_{i, p_i-1}}$ are $p_i$ places of $F_{i-1}$ lying above $P_{i, \infty}$. Then $\bar{\tau}_i$, as an automorphism of $F_{i-1}$, permutes these $p_i$ places. Without loss of generality, assume
\begin{equation}\label{eq:cycle2}
P_{i-1, \infty}\stackrel{\bar{\tau}_i}{\longrightarrow}P_{\lambda_{i, 1}}\stackrel{\bar{\tau}_i}{\longrightarrow}\dots \stackrel{\bar{\tau}_i}{\longrightarrow}P_{\lambda_{i, {p_i-1}}}\stackrel{\bar{\tau}_i}{\longrightarrow}P_{i-1, \infty},\end{equation}
Then we have
\[x_i=\frac{u_{i, i-1}(x_{i-1})}{(x_{i-1}-\lambda_{i,1})\cdots (x_{i-1}-\lambda_{i,|p_i|-1})},\]
for some polynomial $u_{i, i-1}(T) \in \F_q[T]$. As $[\F_q(x_{i-1}): \F_q(x_i)]=p_i$, by Lemma~\ref{lem:simple}, we then have $\deg u_{i,i-1}=p_i$. This completes the proof of (2).

(3) For the relationship between $y_i$ and $x_{i-1}$, we note that
\[\begin{split}
y_i=\prod_{\bar{\tau}\in G_{i}/G_{i-1}} \bar{\tau}(y_{i-1})=\prod_{\bar{\tau}\in G_{i}/G_{i-1}} \bar{\tau}(1/Q_{i-1}(x_{i-1}))\end{split}.\]
The second equality in the above display follows from (1). 
Consider the principal divisor $\ddiv(1/(Q_{i-1}(x_{i-1})))=2P_{{i-1}, \infty}-Q_{i-1}$ in $F_i$. By Equation~\eqref{eq:cycle2},
\[\ddiv(y_i)=\sum_{j=0}^{p_i-1}\ddiv\left( \bar{\tau}_i^{j}(1/Q_{i-1}(x_{i-1}))\right)=2P_{i,\infty}+\sum_{j=0}^{p_i-1}(2P_{\lambda_{i, j}})-p_iQ_{i-1}.\]
Thus $y_i=c_i\cdot\frac{(x_{i-1}-\lambda_{i,1})^2\cdots (x_{i-1}-\lambda_{i,|p_i|-1})^2}{Q_{i-1}(x_{i-1})^{p_i}}=c_i\cdot y_{i-1}^{p_i}\cdot(x_{i-1}-\lambda_{i,1})^2\cdots (x_{i-1}-\lambda_{i,|p_i|-1})^2$ for some $c_i\in\F_q^*$. 

(4) can be proved similarly by substituting $i$ with $r$ in the proof of (2) and (3). We do not repeat it here.
\end{proof}

\begin{remark}
If $n=q+1=2^r$, by Equations \eqref{eq:usfx} and \eqref{eq:usfy} in the proof of Lemma~\ref{lem:key lemma}, then
\begin{equation}\label{eq:xyr}
\begin{split}
&x_r=\frac{u_{r,0}(x)}{x^q-x}\ \text{for\ some}\ u_{r,0}(x)\in\F_q[x]\ \text{with}\ \deg(u_{r,0}(x))=q+1,\\
&y_r=\frac{(x^q-x)^2}{c_{r,0}Q^{q+1}(x)}\ \text{for\ some}\ c_{r,0}\in\F_q^*.
\end{split}
\end{equation}
\end{remark}

In the following, we will first construct a basis $\cB$ of the Riemann-Roch space $\cL(nQ)$.
For simplicity, we define the following notions:
\begin{align}\label{eq: defzk}
\begin{split}
&\text{for}\ 0\leq i\leq r-1,\ z_i^{(e_i)}:=\begin{cases}
1,\ &\text{if}\ e_i=0,\\
 \frac 1{(x_{i}-\lambda_{i+1,1})(x_{i}-\lambda_{i+1,2})\cdots(x_{i}-\lambda_{i+1, e_i})},\ &\text{if}\ 1\leq e_i\leq p_{i+1}-1,\end{cases}\\
 &z_r^{(0)}:=1,\ z_r^{(1)}:=x_r,\ \text{and}\ p_{r+1}:=2.\\
 \end{split}
\end{align}
where all $(x_{i}-\lambda_{i+1, e_i})$ are given in Lemma~\ref{lem:key lemma}(2) for $0\leq i\leq r-1$ and $e_i\in\Z_{p_{i+1}}$.

\begin{lemma}\label{lem:basisq+1}
We continue the notations in Lemma~\ref{lem:key lemma}. Let $z_i^{(e_i)}$ be defined as in Equation~\eqref{eq: defzk} for $i=0, 1, \ldots, r$ and $e_i\in \Z_{p_{i+1}}$. Then
\[
\cB=\left\{z_{0}^{(e_{0})}z_{1}^{(e_1)}\cdots z_{r-1}^{(e_{r-1})}z_r^{(e_r)}y_r \mid \be=(e_0,e_1,\ldots,e_{r-1}, e_{r})\in\prod_{i=1}^{r+1}\Z_{p_i}\right\}\cup \left\{1\right\}
\]
is a basis of the Riemann-Roch space $\cL\left(nQ\right)$.
\end{lemma}
\begin{proof}
For $m=1, 2,\dots, r$, define the set $\cB_{m-1}$ as follows
\begin{equation}\label{eq:mbasis}
\cB_{m-1}=\underbrace{\{y_m, y_m^2,\dots, y_m^{\frac n{p_1\cdots p_{m}}}\}\cdot\{1, x_{m}\}}_\text{$\cB_{m-1,m}$}\cdot\underbrace{\{1,z_{m-1}^{(1)},\ldots,z_{m-1}^{(p_{m}-1)}\}}_\text{$\cB_{m-1,m-1}$},
\end{equation}
where the product ``$\cdot$" of sets means a set consisting of all possible products of elements
from these three sets. Then $|\cB_{m-1}| = 2n/(p_1\cdots p_{m-1})$. Let $\cB_{m-1,m}$ and $\cB_{m-1,m-1}$ be two sets defined in Equation~\eqref{eq:mbasis}. We claim that
 \begin{itemize}
\item[(i)] $\cB_{m-1}\cup \{1\}$ is a basis of the Riemann-Roch space $\cL\left(\frac{n}{p_1\cdots p_{m-1}}Q_{m-1}\right)$.
\item[(ii)] $\cB_{m-1,m}\cup\{1\}$ is a basis of the Riemann-Roch space $\cL\left(\frac{n}{p_1\cdots p_{m}}Q_{m}\right)$.
\end{itemize}
Then $\cB_{m}$ and $\cB_{m-1,m}$ are linearly equivalent and have same cardinality. By substituting the subset $\cB_{m-1,m}$ with $\cB_{m}$ in $\cB_{m-1}$, we get a new basis $\cB_{m-1}'\cup\{1\}$ of $\cL\left(\frac{n}{p_1\cdots p_{m-1}}Q_{m-1}\right)$, where $\cB_{m-1}'=\cB_{m}\cdot\cB_{m-1, m-1}$. In particular, for $m=0$, $\cB_0\cup\{1\}=(\cB_{0, 1}\cdot\cB_{0, 0})\cup\{1\}$ is a basis of $\cL\left(nQ\right)$. By substituting $\cB_{0, 1}$ with $\cB_{1}$, then we have $(\cB_{1}\cdot\cB_{0, 0})\cup\{1\}$ is also a basis of $\cL(nQ)$. Continuing in this fashion till to $m={r-1}$, then we get a basis of $\cL(nQ)$, i.e.,
\[\begin{split}\cB&=\cB_{r}\cdot \cB_{r-1,r-1}\cdots\cB_{1,1}\cdot\cB_{0,0}\cup\{1\}\\
=\{y_r, &x_ry_r\} \cdot\{1,z_{r-1}^{(1)},\ldots,z_{r-1}^{(p_r-1)}\}\cdots \{1,z_{0}^{(1)},\ldots,z_{0}^{(p_1-1)}\} \cup\{1\}.
\end{split}\]
Therefore, it suffices to prove the above claims (i) and (ii).

(i) For every element $y_m^ix_m^jz_{m-1}^{(e_{m-1})}$ in $\cB_{m-1}$, where $1\leq i\leq n/{p_1\cdots p_m}$, $j=0,1$ and $e_{m-1}\in\Z_{p_m}$, by Lemma~\ref{lem:key lemma}, we have
\[\nu_{Q_{m-1}}(y_m^ix_m^jz_{m-1}^{(e_{m-1})})=\nu_{Q_{m-1}}\left(y_m^i\right)=-i\cdot p_m\geq -n/(p_1\cdots p_{m-1}),\]
and it has no other poles except for $Q_{m-1}$. Thus $\cB_{m-1}\cup \{1\}\subsetneq\cL\left(\frac{n}{p_1\cdots p_{m-1}}Q_{m-1}\right)$. Assume
\begin{equation}\label{eq:usfi}
a+\sum_{i=1}^{n/(p_1\cdots p_m)}\sum_{j=0}^{1}\sum_{e=0}^{p_m-1} a_{i,j,e} \cdot y_m^ix_m^jz_{m-1}^{(e)}=0
\end{equation}
where all $a_{i,j,e}\in \F_q$ and $a\in \F_q$. Multiplying by $(x_{m-1}-\lambda_{m, 1})\cdots (x_{m-1}-\lambda_{m, p_m-1})$ on both sides of Equation~\eqref{eq:usfi}, we have
\begin{equation}\label{eq:usfii}
\begin{split}
&\left(a+\sum_{i=1}^{n/(p_1\cdots p_m)}\sum_{j=0}^{1} a_{i,j,0} \cdot y_m^ix_m^j\right)(x_{m-1}-\lambda_{m, 1})\cdots (x_{m-1}-\lambda_{m, p_m-1})\\
&+\sum_{i=1}^{n/(p_1\cdots p_m)}\sum_{j=0}^{1}\sum_{e>0}a_{i,j,e}y_m^ix_m^j(x_{m-1}-\lambda_{m, e+1})\cdots (x_{m-1}-\lambda_{m, p_m-1})=0\end{split}
\end{equation}
Since $[F_{m-1}: F_m]=p_m$, we have the coefficients of $x_{m-1}^{k}$ for $0\leq k\leq p_m-1$ are all zero; namely, $a+\sum_{i,j} a_{i,j,0} \cdot y_m^ix_m^j=0$ and $\sum_{i,j} a_{i,j,e} \cdot y_m^ix_m^j=0$ for every $0<e<p_m$.  As $[\F_q(x_m): \F_q(y_m)]=2$, we then have $a=0$ and $\sum_{i} a_{i,j,e} \cdot y_m^i=0$ for all $0\leq e\leq p_m-1$ and $j=0, 1$. As $y_m$ is transcendental over $\F_q$, we have all the coefficients $a_{i,j,e}=0$. Therefore, the elements in $\cB_{m-1}$ are linearly independent. By the Riemann-Roch theorem \cite{sti09}, $\cL\left(\frac{n}{p_1\cdots p_{m-1}}Q_{m-1}\right)$ has dimension:
\[\ell\left(\frac{n}{p_1\cdots p_{m-1}}Q_{m-1}\right)=2n/(p_1\cdots p_{m-1})+1=|\cB_{m-1}\cup \{1\}|.\]
 Thus, $\cB_{m-1}\cup \{1\}$ is a basis of $\cL\left(\frac{n}{p_1\cdots p_{m-1}}Q_{m-1}\right)$.

(ii) For every element $y_m^ix_m^j$ in the set $\cB_{m-1,m}=\{y_m, y_m^2,\dots, y_m^{\frac n{p_1\cdots p_{m}}}\}\times\{1, x_{m}\}$, by Lemma~\ref{lem:key lemma}, we have \[\nu_{Q_{m}}\left( y_m^ix_m^j \right)=\nu_{Q_{m}}( y_m^i)=-i\geq- n/(p_1\cdots p_m).\]
Thus $\cB_{m-1,m}\subsetneq\cL\left(\frac{n}{p_1\cdots p_{m}}Q_{m}\right)$. We have shown in (i) that $\cB_{m-1, m}\cup \{1\}$ is a linearly independent set. By the Riemann-Roch theorem \cite{sti09}, $\ell(\frac{n}{p_1\cdots p_{m}}Q_{m})=2n/{(p_1\cdots p_m)}+1=|\cB_{m-1,m}\cup\{1\}|$. Hence $\cB_{m-1,m}\cup\{1\}$ is also basis of $\cL\left(\frac{n}{p_1\cdots p_{m}}Q_{m}\right)$.
\end{proof}

Next, we take a subset of $\cB$ to construct a basis of $\F_q[x]_{<n}$.
\begin{lemma}\label{lem:subbasisq+1}
Let $z_i^{(e_i)}$ be defined as in Equation~\eqref{eq: defzk} for $0\leq i\leq r$ and $e_i\in \Z_{p_{i+1}}$. Let $x_r=\frac{u_{r,0}(x)}{\prod_{\alpha\in\hcalP}(x-\alpha)}$ by Lemma~\ref{lem:key lemma} for $u_{r,0}(x)\in\F_q[x]$ of degree $n$. Then
\[\tcB=\{z_0^{(e_0)}z_1^{(e_1)}\cdots z_{r-1}^{(e_{r-1})}x_ry_r \mid \be=(e_0,e_1,\ldots,e_{r-1})\in\Z_{p_1}\times\Z_{p_2}\times\ldots\times\Z_{p_r}\}\]
is a basis of $\frac {u_{r,0}(x)}{Q^n(x)}\F_q[x]_{<n}$. In other words, 
$\frac{Q^n(x)}{u_{r,0}(x)}\tcB$ is a basis of $\F_q[x]_{<n}$.
\end{lemma}
\begin{proof}
Firstly, let $W_q(\tcB)$ be the $\F_q$-subspace of $\cL(nQ)$ spanned by $\tcB$. By Lemma~\ref{lem:basisq+1}, $\tcB$ is an $\F_q$-linearly independent subset of $\cB$. Then
\[\dim(W_q(\tcB))=p_1p_2\cdots p_r=n=\dim\left(\frac {u_{r,0}(x)}{Q^n(x)}\cdot\F_q[x]_{<n}\right).\]
It suffices to show that $\tcB\subset \frac {u_{r,0}(x)}{Q^n(x)}\cdot\F_q[x]_{<n}$.

By Lemma~\ref{lem:key lemma},
 \[z_r^{(1)}y_r =x_ry_r=\frac{u_{r,0}(x)}{\prod_{\alpha\in\hcalP}(x-\alpha)}\cdot \frac{\prod_{\alpha\in\hcalP}(x-\alpha)^2}{Q^n(x)}=\frac {u_{r,0}(x)}{Q^n(x)}{\prod_{\alpha\in\hcalP}(x-\alpha)}.\] 
\textbf{Claim}: assume $z_0^{(e_0)}z_1^{(e_1)}\cdots z_{r-1}^{(e_{r-1})}=N_{\be}(x)/D_{\be}(x)\in \F_q(x)$, where $N_{\be}(x), D_{\be}(x)\in \F_q[x]$. Then the denominator $D_{\be}(x)$ will be killed by ${\prod_{\alpha\in\hcalP}(x-\alpha)}$. 

By the above claim, we have $z_0^{(e_0)}z_1^{(e_1)}\cdots z_{r-1}^{(e_{r-1})}z_r^{(1)}y_r\in\frac{u_{r,0}(x)}{Q^n(x)}\F_q[x]$. Since $W_q\left(\cB\setminus\{1\}\right)=\frac{1}{Q^n(x)}\F_q[x]_{< 2n}$, then the numerator of $z_0^{(e_0)}z_1^{(e_1)}\cdots z_{r-1}^{(e_{r-1})}z_r^{(1)}y_r$ (as a rational polynomial function) belongs to $u_{r,0}(x)\F_q[x]_{< n}$. Thus, the lemma is proved if the claim is true.

\textbf{The proof of claim:} It is equivalent to show that the pole places of $z_0^{(e_0)}z_1^{(e_1)}\cdots z_{r-1}^{(e_{r-1})}$ is a subset of $\calP$. Since $z_{i}^{(e_i)}={1}/{(x_{i}-\lambda_{i+1,1})\cdots (x_{i}-\lambda_{i+1,e_i})}$, the pole places of $z_{i}^{(e_i)}$ in $\PP_{F}$ are those lying above $P_{\lambda_{i+1,k}}$ for $1\leq k\leq e_i$, where $P_{\lambda_{i+1,k}}\in \PP_{F_i}$ is the zero of  $x_{i}-\lambda_{i+1,k}$. Moreover, $P_{\lambda_{i+1,k}}\mid P_{i+1,\infty}$ and $P_{i+1,\infty}$ splits completely in $F/F_{i+1}$ by Lemma~\ref{lem:key lemma}, thus each $P_{\lambda_{i+1,k}}$ splits completely in $F/F_i$. Let $\left(z_{i}^{(e_i)}\right)_{\infty}\in \Div(F)$ be the pole place of $z_{i}^{(e_i)}$ and $Z_i^{(e_i)}:=\Supp\left(\left(z_{i}^{(e_i)}\right)_{\infty}\right)$ be its support set. Then $|Z_i^{(e_i)}|=e_i\cdot p_1\cdots p_i$. In particular, $Z_i^{(p_{i+1}-1)}$ contains $Z_i^{(e_i)}$ which has cardinality $\left(\prod_{k=1}^{i+1} p_k\right)-\left(\prod_{k=1}^{i}p_k\right)$. Moreover, let $\calP_{i}:=\{P\in\PP_F:\ P\mid P_{i,\infty}\}$. Then $Z_i^{(p_{i+1}-1)}\cap\calP_i=\emptyset$. For any $j< i$, since $Z_j^{(e_j)}\subset Z_j^{(p_{j+1}-1)}\subset \calP_i$, $Z_j^{(e_j)}$ and $Z_i^{(e_i)}$ are disjoint. This means the disjoint union $\bigsqcup_{i=0}^{r-1} Z_i^{(e_i)}$ is a subset of $\bigsqcup_{i=0}^{r-1} Z_i^{(p_{i+1}-1)}=\hcalP$. Thus, the denominator $D_{\be}$ can be killed by ${\prod_{\alpha\in\hcalP}(x-\alpha)}$. 

\end{proof}

Assume  $n=\prod_{i=1}^rp_i$ is $B$-smooth, i.e., $p_i<B$ for all $i=1,\ldots,r$. By Lemma~\ref{lem:q+1subf}, there are $\frac{q+1}n$ rational places in $\PP_{F_r}$ that are splitting completely in $F_0/F_r$, namely those places lying above the pole place of $\sum_{i=0}^{q}\sigma^i(x)$ in $F^{\langle \sigma\rangle}$. Denote the set consisting of these $\frac{q+1}n$ places by $\calQ$. Particularly, $\calQ=\{P_{r,\infty}\}$ if $n=q+1$, and contains $P_{r,\infty}$ if $n<q+1$. We make the following choices: 
\begin{equation}\label{eq:fP}
\begin{split}
&\fP\in \calQ\setminus \{P_{r,\infty}\},\ \text{if}\ n<q+1,\\
&\fP=P_{r,\infty},\ \text{if}\ n=q+1.\\
\end{split}
\end{equation}
Let $\calP\subset\PP_F$ be the set of all rational places of $F$ lying over $\fP$ (hence $|\calP|=n$). If $n<q+1$, by the proof of claim in Lemma~\ref{lem:subbasisq+1}, every base element $z_0^{(e_0)}z_1^{(e_1)}\cdots z_{r-1}^{(e_{r-1})}x_ry_r\in\tcB$ has no poles in $\calP$. 

\begin{theorem}\label{thm: q+1-case}
Let $G$ be a cyclic subgroup of $\langle \sigma \rangle$ of order $n=\prod_{i=1}^rp_i$. Let $\fP\in\PP_{F^G}$ and $\calP$ be defined as above. If $n$ is $B$-smooth, then the DFT of any polynomial $f\in\F_q[x]_{<n}$ at $\calP$ can be done in $O(Bn\log n)$ operations of $\F_q$. 
\end{theorem}
\begin{proof}
Let $\tilde{f}=\frac{u_{r,0}(x)}{c_{r,0}Q^{n}(x)}f$, where $u_{r,0}(x)$ and $c_{r,0}Q^n(x)$ come from $x_r=\frac{u_{r,0}(x)}{\prod_{\Ga\in\hcalP}(x-\Ga)}$ and $y_r=\frac{\prod_{\Ga\in\hcalP}(x-\Ga)^2}{c_{r,0}Q^n(x)}$ by Lemma~\ref{lem:key lemma}(4). Then $\tilde{f}\in W_q(\tcB)$ by Lemma~\ref{lem:subbasisq+1}. We will first show that the MPE of $\tilde{f}$ at $\calP$ can be done in $O(n\log n)$. By pre-computing $\frac{c_{r,0}Q^{n}(\Ga)}{u_{r,0}(\Ga)}$ for all $\Ga\in\calP$, we can obtain $f(\Ga)=\frac{c_{r,0}Q^{n}(\Ga)}{u_{r,0}(\Ga)}\tilde{f}(\Ga)$ for every $\Ga\in\calP$. If $n=q+1$, we define $f(P_{\infty})=\infty$ since $\nu_{P_{\infty}}(f)<0$. 
Denote the complexity (i.e., the total number of operations of $\F_q$) to compute the $n$ evaluations of any $\tilde{f}\in W_q(\tcB)\subset\cL(nQ)$ at the set $\calP$ by $C(n)$. 

Under the basis $\tcB$ in Lemma~\ref{lem:subbasisq+1}, we represent $\tilde{f}$ as
\begin{equation}\label{eq:fft}
\begin{split}
\tilde{f}&=\sum_{\bbe\in\prod_{i=0}^{r-1} \Z_{p_{i+1}}} a_{\bbe} \cdot {x_ry_r\bz^{(\bbe)}}\\
&=\sum_{e_0\in\Z_{p_1}}z_0^{(e_0)} \cdot\sum_{\bbe'\in\prod_{i=1}^{r-1}\Z_{p_{i+1}}} a_{\underline{\be'e_0}} \cdot {x_ry_r\bz^{(\bbe')}} \\
&=z_0^{(0)}\cdot f_0+z_0^{(1)}\cdot f_1+\dots+z_0^{(p_1-1)}\cdot f_{p_1-1}\\
&=f_0+\frac 1{x-\lambda_{1,1}}\cdot f_1+\dots+\frac 1{(x-\lambda_{1,1})\cdots(x-\lambda_{1,p_1-1})}\cdot f_{p_1-1}.\end{split}
\end{equation}
where $f_{k}=\sum_{\bbe'\in\prod_{i=1}^{r-1}\Z_{p_{i+1}}} a_{\underline{\be'k}} \cdot{x_ry_r\bz^{(\bbe')}}$ for every $0\leq k<p_1$. Now, all $f_{k}\in F_1$ and $$\nu_{Q_1}(f_k)\geq \nu_{Q_1}(y_r)=-\frac n{p_1},\ \text{where}\ Q_1=Q\cap F_1.$$
Thus, $f_{k}\in \cL(\frac n{p_1}Q_1)$ for all $0\leq k\leq p_1-1$.

For any $P\in \calP$, let $P_1=P\cap F_1$. Then $f_k(P)=f_k(P_1)$ as $f_k\in F_1$. Let $\calP^{(1)}=\calP\cap F_1=\{P\cap F_1: P\in \calP\}$. Then $|\calP^{(1)}|=\frac{n}{p_1}$ by Lemma~\ref{lem:q+1subf}. Thus the evaluations of $\tilde{f}$ at $\calP$ can be reduced to evaluations of $f_0, \cdots, f_{p_1-1}$ at $\calP^{(1)}$, and then a combination of these values according to Equation~\eqref{eq:fft}. If $n<q+1$, by the choice of $\fP$, there are no places in $\calP$ that are poles of  $1/{(x-\lambda_{1,1})\cdots(x-\lambda_{1,p_1-1})}$. So the reduction works well and we can directly get the recursive formula~\eqref{eq:recq+1} which leads to the final result $C(n)=O(Bn\log n)$. However, if $n=q+1$, the pole places of $1/{(x-\lambda_{1,1})\cdots(x-\lambda_{1,p_1-1})}$ are contained in $\calP$ which need to be dealt with separately. 

We are left to consider the case $n=q+1$. According to Lemma~\ref{lem:key lemma} (2), let $\Lambda_{1}=\{P_{\lambda_{1,1}},\dots,P_{\lambda_{1,p_1-1}}, P_{0,\infty}\}$ be the set of all rational places of $F_0$ lying above $P_{1,\infty}$. For any $P_{\lambda_{1,i}}\in \Lambda_1$, by computations,
\begin{equation}\label{eq:valus}
\begin{split}
&\nu_{P_{\lambda_{1,i}}} (y_r)=\nu_{P_{r,\infty}} (y_r)=2,\ \nu_{P_{\lambda_{1,i}}} (x_r)=-1,\\
&\nu_{P_{\lambda_{1,i}}} (z_j^{(e_j)})=\nu_{P_{j,\infty}} (z_j^{(e_j)})=e_j\geq 0,\ \text{for}\ 1\leq j\leq r-1,\ e_j\in\Z_{p_{j+1}}\\
&\nu_{P_{\lambda_{1,i}}} (z_0^{(k)})=-1,\ \text{for}\ i\leq k\leq p_1-1,\\
 &\nu_{P_{\lambda_{1,i}}} (z_0^{(k)})=0,\ \text{for}\ 1\leq k< i,\\
\end{split}\end{equation}
Then, \[\nu_{P_{\lambda_{1,i}}} (f_k)\geq \min_{\bbe'\in\prod_{i=1}^{r-1}\Z_{p_{i+1}}}\{\nu_{P_{\lambda_{1,i}}} (a_{\underline{\be'k}} \cdot{x_ry_r\bz^{(\bbe')}})\}\geq 1,\]
namely, each $P_{\lambda_{1,i}}$ is a zero of $f_k=\sum_{\bbe'\in\prod_{i=1}^{r-1}\Z_{p_{i+1}}} a_{\underline{\be'k}} \cdot{x_ry_r\bz^{(\bbe')}}$. However, $\{P_{\lambda_{1,1}},\dots,P_{\lambda_{1,k}}\}$ are simple poles of $z_0^{(k)}$. We cannot evaluate $f_k$ and $z_0^{(k)}$ separately at $P_{\lambda_{1,i}}$ for $i\leq k$.

Note that, by Equation~\eqref{eq:valus}, $\nu_{P_{\lambda_{1,i}}} ({x_ry_r\bz^{\bbe})}\geq 0$ and take $0$ only at $\bbe=(0\cdots 0k)$ for $i\leq k\leq p_1-1$. Thus
\begin{small}
\begin{equation}\label{eq:fft1}
\begin{split}
{z_0^{(k)}f_k}(P_{\lambda_{1,i}})&=\big(a_{0\cdots 0k}x_ry_r\cdot z_0^{(k)}\big)(P_{\lambda_{1,i}})\\
&=\left(a_{0\cdots 0k}\frac{u_{r,0}(x)\prod_{\alpha\in\hcalP}(x-\alpha)}{Q^n(x)}\cdot \frac{1}{(x-\lambda_{1,1})\cdots(x-\lambda_{1,k})}\right)(P_{\lambda_{1,i}})\\
&=\begin{cases}
a_{0\cdots 0k}\cdot\frac{u_{r,0}(\lambda_{1,i})\prod_{\alpha\in\hcalP\setminus\{\lambda_{1,1},\ldots,\lambda_{1,k}\}}(\lambda_{1,i}-\alpha)}{Q(\lambda_{1,i})^{n}},\  & \text{if}\ i\leq k\leq p_1-1;\\
0,\ & \text{if}\ k< i\leq p_1-1.\end{cases}
\end{split}
\end{equation}
\end{small}
where the second ``=" follows from Lemma~\ref{lem:key lemma} (4).

Thus, for these $p_1-1$ places in $\Lambda_{1}$, we can get $\tilde{f}(P_{\lambda_{1,i}})$ according to equations~\eqref{eq:fft1} and \eqref{eq:fft} which can be done in $p_1^2$ by pre-computing of
\[\frac{u_{r,0}(\lambda_{1,i})\prod_{\alpha\in\hcalP\setminus\{\lambda_{1,1},\ldots,\lambda_{1,k}\}}(\lambda_{1,i}-\alpha)}{Q(\lambda_{1,i})^{n}}\ \text{for}\ 1\leq i\leq k\leq p_1-1.\]
It is not hard to see $\nu_{P_{0,\infty}} (z_0^{(k)}f_k)>0$, thus $z_0^{(k)}f_k(P_{0,\infty})=0$. As for other $n-p_1$ rational places in $\calP\setminus \Lambda_{1}$, the functions $\{z_0^{(k)}, f_k\}_{k=0}^{p_1-1}$ in equation~\eqref{eq:fft} are well-defined at them. Thus, we can get the corresponding $n-p_1$ evaluations of $\tilde{f}$ according to Equations~\eqref{eq:fft} by using $p_1\cdot O(n-p_1)$ operations in $\F_q$. Since each $f_k\in \cL(\frac{n}{p_1}Q_1)$ and $|\calP_1|=n/p_1$, by the above analysis, we get a recursive formula for $C(n)$:
\begin{equation}\label{eq:recq+1}
C(n)=p_1\cdot C\left({n}/{p_1}\right)+p_1\cdot O(n).
\end{equation}
Then we can continue to reduce the DFT of each $f_k$ at $\calP^{(1)}$ to the DFTs of $p_2$ functions in $\cL\big((n/{p_1p_2})Q_2\big)$ at $\calP^{(2)}=\calP\cap F_2$, till to the last step where we get $n$ functions in $\cL(Q_r)$ and evaluate them at $\calP^{(r)}=\{P_{r,\infty}\}$, which can be done in $O(n)$ operations. Overall, the recursive formula~\eqref{eq:recq+1} lead to a final complexity of $C(n)$:
\[\begin{split}
C(n)=p_1p_2\cdots p_r\cdot C(1)+(p_r+\dots+p_2+p_1)\cdot O(n)=O(B\cdot n\log n).\end{split}\]
Since we need some precomputations in each step of the reduction, we analyze the total storage for these precomputations. Precisely, in step $j$ of the reduction, we need to precompute
\[\left(y_rx_rz_j^{(k)}\right)(P_{\lambda_{j,i}})=\frac{u_{r,j}(\lambda_{j,i})\prod_{\alpha\in\hcalP^{(j)}\setminus\{\lambda_{j,1},\ldots,\lambda_{j,k}\}}(\lambda_{j,i}-\alpha)}{Q_j(\lambda_{j,i})^{n}}\ \text{for}\ 1\leq i\leq k\leq p_j-1.\]
Therefore, the storage for the precomputations in the process of reduction is:
\[\sum_{j=1}^r\sum_{k=1}^{p_j-1}\sum_{i\leq k} 1=O(B^4\log n).\]
By adding the storage for $\frac{u_{r,0}(\alpha)}{Q^{n}}(\Ga)$ for every $\Ga\in\calP$, the total storage for precomputation is $n+O(B^4\log n)=O(n)$.

\end{proof}

\begin{remark}
\begin{itemize}

\item[(1)] By the same analysis as in Remark~\ref{rmk:ifft}, the inverse G-FFT for the $q+1$ case can also be performed in $O(n\log n)$ operations.

\item[(2)] Let $f(x)\in\F_q[x]_{<n}$ be any nonzero polynomial. In the affine G-FFT, we showed that the transformation from the coefficients of $f(x)$ under the basis $\cB$ in Lemma~\ref{lem: new basis} to the coefficients of $f(x)$ under the standard basis can be done in at most $O(n\log n)$ operations.  However, in the $q+1$ case, this transformation will cost $O(M(n)\log n)$ operations in $\F_q$, where $M(n)$ denotes the cost of multiplication of two polynomials of degree less than $n$. For simplicity, we take $q+1=2^r$ as an example to illustrate. 
In this case, assume $x_{i+1}=\frac{u_{i,i+1}(x_i)}{x_i-\lambda_i}$ for $i=0,1,\ldots,r-1$. Assume $f(x)=(x^q-x)\cdot\sum_{\bbe\in \Z_2^{r}} a_{\bbe} \bz^{(\bbe)}$, where each $\bz^{(\bbe)}=z_0^{(e_0)}z_1^{(e_1)}\cdots z_{r-1}^{(e_{r-1})}$ is the base element in Lemma~\ref{lem:subbasisq+1}, and the evaluation set $\calP=\calP_{i,0}\sqcup \calP_{i,1}$, where $i=1,2,\ldots,r-1$ and $\calP_{i,0}$ (resp. $\calP_{i,1}$) $\subsetneq \PP_{F}$ is the set of places lying above $P_{i,\infty}$ (resp. $P_{i,\lambda_i}$). Since 
\[
\begin{split}
f(x)&=(x^q-x)\cdot\sum_{\bbe'\in \Z_2^{r-1}} a_{\underline{\be'0}}\bz^{(\bbe')}+(x^q-x)z_{r-1}^{(1)}\cdot\sum_{\bbe'\in \Z_2^{r-1}} a_{\underline{\be'1}}\bz^{(\bbe')}\\
&=f_0(x)\cdot\prod_{\alpha\in\calP_{r-1,1}}(x-\alpha)+f_1(x).
\end{split}\]
where \[f_0(x)=\prod_{\alpha\in\hcalP_{r-1,0}}(x-\alpha)\sum_{\bbe'\in \Z_2^{r-1}} a_{\underline{\be'0}}\bz^{(\bbe')},\ 
f_1(x)=\prod_{\alpha\in\hcalP_{r-1,0}}(x-\alpha)\sum_{\bbe'\in \Z_2^{r-1}} a_{\underline{\be'1}}\bz^{(\bbe')}.\] The pole places of $\bz^{(\bbe')}$ are in $\hcalP_{r-1,0}$ for all $\bbe'\in \Z_2^{r-1}$. Therefore, the computation of coefficients of $f(x)$ under the standard basis can be reduced to the computations of $f_0(x)$ and $f_1(x)$ under the standard basis. Then get $f$ by multiplying $f_0$ with $\prod_{\alpha\in\calP_{r-2,1}}(x-\alpha)$ and add $f_1$. Thus the transformation complexity, denoted by $B(n)$, satisfies 
\[B(n)=2B(n/2)+M(n)=O(M(n)\log n).\]
\end{itemize}
\end{remark}

\begin{example}
At last, let us present an example with $q+1=2^r$ to illustrate the G-FFT. Since $x_i=\frac{u_{i,i-1}(x_{i-1})}{x_{i-1}-\lambda_{i-1}}$ for $i=1,\dots,r$ by Lemma~\ref{lem:key lemma}, the pole place $P_{i,\infty}\in\PP_{F_i}$ of $x_i$ splits into two places in $\PP_{F_{i-1}}$, i.e., $P_{i-1,\infty}$ and $P_{i-1,\lambda_{i-1}}$. Let $c_i=\frac{y_rx_r}{x_i-\lambda_i}\mid_{x_i=\lambda_i}$ for $i=0,1,\ldots,r$. Then the G-FFT in the proof of Theorem~\ref{thm: q+1-case} can be implemented by the following Algorithm~1. Note that it can be also applied to the general case, i.e., $q+1=\prod_{i=1}^r p_i$ is $B$-smooth, by writing $\tilde{f}$ under the basis $\tcB$ in step 4 and adding the discussion of $p_{j}-1$ pole places $\{P_{j,1},\ldots,P_{j,p_j-1}\}\subset \calP^{(j-1)}$ of $x_{j}$ in step 8. For better understanding, we only present the simplest case $q+1=2^r$.
 \begin{algorithm}[!h]\label{alg: q+1}
  \caption{ \bf G-FFT($f, \calP$).}
  \label{alg:Framwork}
  \begin{algorithmic}[1]
    \Require
     $f=\sum_{\bbe\in\Z_{2}^r}a_{\bbe}\frac{x^q-x}{\prod_{i=0}^{r-1}(x_i-\lambda_i)^{e_i}}\in \F_q[x]_{<q}$ and $\calP=\F_q=\{P_{1},P_2,\dots,P_q\}$.
    \Ensure
      $\left(f(P_{1}),\dots,f(P_{q})\right)\in\F_q^{q}$.
      \begin{itemize}
   \item[--] Write $\tilde{f}=\frac{u_{r,0}(x)}{c_{r,0}Q^{q+1}(x)}f=\sum_{\bbe\in\Z_{2}^r}a_{\bbe}\frac{y_rx_r}{\prod_{i=0}^{r-1}(x_i-\lambda_i)^{e_i}}$.
    \item[--] For $j=0, 1,\dots,r$, let $\calP^{(j)}=\calP\cap F_j$. In particular, $\calP^{(0)}=\calP$.\end{itemize}
    \For {$j=0$ to $r-1$}
          \State Write $\tilde{f}=f_{0}+f_{1}\frac{1}{x_j-\lambda_j}$.  
           \State $j\gets j+1$
          \State Call G-FFT$\left(f_{k},\calP^{(j)}\right)$ for $k=0,1$.
          \State Return $\tilde{f}(P)=f_{0}(P)+f_{1}(P)\cdot \frac{1}{x_{j-1}(P)-\lambda_{j-1}}$ for $P\in \calP^{(j-1)}\setminus\{P_{\lambda_{j-1}}\}$;
          \State Return $\tilde{f}(P)=a_{\underline{0\cdots 01}}\cdot c_{j-1}$ for $P=P_{\lambda_{j-1}}$, where $a_{\underline{0\cdots 01}}$ is the coefficient of $f$ at the term $\frac{y_rx_r}{x_{j-1}-\lambda_{j-1}}$.
     \EndFor
   \State Output $f(P)=\frac{c_{r,0}Q^{q+1}(P)}{u_{r,0}(P)}\tilde{f}(P)$ for every $P\in\calP$.
  \end{algorithmic}
\end{algorithm}

More specifically, all data for the case $q=127$ are as follows. Let $F=\F_{127}(x)$. Firstly, we choose a primitive polynomial over $\F_{127}$: $m(x)=x^2 +126x +3$. The automorphism $\sigma=\left(\begin{matrix}0&1\\ 124&1\end{matrix}\right)\in\PGL_2(q)$ has order $128$. Since $q+1=2^7$ and $r=7$, we can construct a tower of fields, i.e., $\F_q(x_i)=F^{G_i}$ is the subfield fixed by $G_i=\langle \sigma^{2^{r-i}} \rangle$, $i=0,1,\ldots,7$. Then $x_i=\sum_{\tau\in G_i}\tau(x)$. In addition, let $Q(x)=x^2+42x+85$ and $Q$ be the quadratic place corresponding to $Q(x)$. Then $Q$ is totally ramified in $F/F^{G}$. Let $y_i=\prod_{\tau\in G_i}\tau(1/Q(x))$. The field $E_i=\F_q(y_i)$ is a subfield of $F_q(x_i)$. By computation, we have

(1) \textbf{The generator $x_i$ (resp. $y_i=1/Q_i(x_i)$) of the subfield $F_i=\F_q(x_i)$ (resp. $E_i=\F_q(y_i)$) for $i=0, 1, \ldots,7$}. By the definition~\eqref{eq:defxy}, we have
\begin{small}
\[\begin{split}
&x_0=x,\ x_1=\frac{x^2+42}{x+21},\ x_2=\frac{x^4 + 125x^2 + 71x + 33}{x^3 + 63x^2 + 28x + 100},\ x_3=\frac{x^8 + 33x^6 + 105x^5 + 24x^4 + \cdots + 122}{x^7 + 20x^6 + 69x^5 + \cdots + 20}\\
&x_4 = \frac{x^{16} + 87x^{14} +34x^{13} + 116x^{12} \cdots+ 43}{x^{15} + 61x^{14} + 91x^{13} + 34x^{12}+\cdots+ 44},\ x_5 = \frac{x^{32} + 4x^{30} + 29x^{29} + 119x^{28}+ \cdots+ 8}{x^{31} + 16x^{30} +22x^{29} + 47x^{28}+\cdots+116},\\
 &x_6 = \frac{x^{64} + 90x^{62} +15x^{61} + 89x^{60}+  \cdots+ 80}{x^{63} + 53x^{62}+67x^{61} + 64x^{60}+ \cdots+71},\ x_7=\frac{x^{128} + 42x + 85}{x^{127} + 126x}=\frac{u_{3,0}(x)}{x^q-x}.
\end{split}\]
\end{small}
and
\begin{small}
\[\begin{split}
&y_0=1/Q(x),\ y_1=\frac{x^2+42x+60}{25Q(x)^2},\ y_2=\frac{x^6 + 126x^5 + \cdots+ 94}{16Q(x)^4},\ y_3=\frac{x^{14} + 40x^{13} + \cdots + 19}{38Q(x)^8},\\
&y_4=\frac{x^{30} + 122x^{29} +\cdots+ 31}{16Q(x)^{16}},\ y_5=\frac{x^{62} + 32x^{61} + \cdots + 121}{100Q(x)^{32}},\ y_6=\frac{x^{126} + 106x^{125} + \cdots + 88}{4Q(x)^{64}},\\
&y_7=\frac{x^{254} + 125x^{128 }+ x^2}{100Q(x)^{128}}.
\end{split}\]
\end{small}

(2) \textbf{The basis of $\F_q[x]_{<128}$}. The pole place $P_{i,\infty}\in\PP_{F_i}$ of $x_i$ splits into two places in $\PP_{F_{i-1}}$, i.e., $P_{i-1,\infty}$ and $P_{i-1,\lambda_{i-1}}$, for every $i$. By computations, we have $\lambda_0= 106, \lambda_1 = 85, \lambda_2 = 43, \lambda_3 = 86, \lambda_4 = 45, \lambda_5 = 90, \lambda_6 = 53$. Thus, by Lemma~\ref{lem:subbasisq+1}, a basis of $\F_q[x]_{<128}$ is
\[\left\{\frac {x^{127}-x}{\prod_{i=0}^6(x_i-\lambda_i)^{e_i}} \mid \bbe=(e_0,e_1,\ldots,e_6)\in\Z_2^7\right\}\]

(3) \textbf{The precomputations}. In the proof of Theorem~\ref{thm: q+1-case}, we need to store $r$ values in precomputing process:
\[\begin{split}
&c_0=\frac{y_7x_7}{x-\lambda_0}\mid_{x=\lambda_0}=106,\ c_1=\frac{y_7x_7}{x_1-\lambda_1}\mid_{x_1=\lambda_1}=101,\ c_2=\frac{y_7x_7}{x_2-\lambda_2}\mid_{x_2=\lambda_2}=64,\\
&c_3=\frac{y_7x_7}{x_3-\lambda_3}\mid_{x_3=\lambda_3}=34,\ c_4=\frac{y_7x_7}{x_4-\lambda_4}\mid_{x_4=\lambda_4}=35,\  c_5=\frac{y_7x_7}{x_5-\lambda_5}\mid_{x_5=\lambda_5}=1,\\
&c_6=\frac{y_7x_7}{x_6-\lambda_6}\mid_{x_6=\lambda_6}=0.\end{split}\]

According to Algorithm 1, we can compute the MPE for any $f=\sum_{\bbe\in\Z_2^7} a_{\bbe} \frac{x^{127}-x}{\prod_{i=0}^6(x_i-\lambda_i)^{e_i}}\in\F_{127}[x]_{<127}$ at $\F_{127}$. In Appendix B, we take a random $f(x)$ (see its coefficients in Table 2) and set $\tilde{f}=\frac{x^{128}+42x+85}{100Q^{128}(x)}f$. Then we compute $\{\tilde{f}(P): P\in\F_{128}\}$ by G-FFT and get $f(P)=\frac{100Q^{128}(P)}{u_{7,0}(P)}\tilde{f}(P)$ for every $P\in\F_{127}$ (see Table~3 for the MPEs of $\tilde{f}$ and $f$ at $\F_q$).


\end{example}


\section{Conclusion}

In this work, we consider FFT for evaluations of polynomials at a chosen set $\calP\subseteq\F_q$.  We conclude that when either $q-1$, $q$ or $q+1$ is smooth, then it is possible to run FFT in $O(n\log n)$ operations of $\mathbb{F}_q$, which loosens the restriction of FFT on the finite field to some extent. Most importantly, if $n\mid (q+1)$ is smooth, we give a practical algorithm to implement FFT over $\F_q$ of length $n$ rather than turn to a quadratic extension field $\F_{q^2}$ to do FFT of length $n$. Our framework is based on the Galois theory and properties of the rational function field. The applications of our new algorithm to polynomial arithmetic and encoding/decoding of Reed-Solomon codes are not considered in this work and we leave them as a future research work.

\section*{Acknowledgments} We would like to thank Yi Kuang for his help in implementing the G-FFT algorithm. We would also like to thank Liming Ma and Fuchun Lin for their valuable comments on the manuscript of this work. This Research is partially supported by the National Key Research and Development Program under Grant 2022YFA1004900 and the National Natural Science Foundation of China under Grant numbers 123031011 and 12031011.

\vskip 1cm
\newcommand{\etalchar}[1]{$^{#1}$}

\appendix

\section{$(x^p-\alpha x)$-adic expansions of polynomials in $\F_q[x]$}

\begin{definition}[\cite{von13, Gao10}]
Let $\F_q$ be a finite field with characteristic $p$. Assume $f(x)\in \F[x]_{<n}$ is a polynomial of degree less than $n$. For any $\alpha\in \F_q$, the $(x^p-\alpha x)$-adic expansion of $f(x)$ is
\[f(x)=a_0(x)+a_1(x)\cdot(x^p-\alpha x)+\dots+a_m(x)\cdot(x^p-\alpha x)^m,\]
where $a_i(x)\in\F_q[x]_{<p}$ for $0\leq i\leq m$ and $m=\lfloor n/p\rfloor$.
\end{definition}
Mateer-Gao \cite{Gao10} showed that when $p=2$, the $(x^2-x)$-adic expansion of $f(x)$ can be computed in $O(n\log n)$. In this part, we will generalize this result to any constant characteristic $p$.

\begin{lemma}\label{lem:poly-adic}
Assume $\F_q$ is a finite field with constant characteristic $p$. Let $f(x)\in \F_q[x]_{<n}$ be a polynomial of degree less than $n$. For any $\alpha\in \F_q$, one can compute the $(x^p-\alpha x)$-adic representation of $f(x)$ in $O(n\log n)$ operations in $\F_q$.
\end{lemma}
\begin{proof}
Let $c_n$ be the leading coefficient of $f(x)$. If $n\leq p$, then the $(x^p-\alpha x)$-adic expansion of $f$ is
\[f(x)=\begin{cases}
f(x), &\text{if}\ \deg(f)<p,\\
c_p(x^p-\alpha x)+ f(x)-c_p(x^p-\alpha x), &\text{if}\ \deg(f)=p.
\end{cases}\]
Next, we assume $\deg(f)>p$. Let $r\geq 2$ be a positive integer such that
\[p^{r-1}< n\leq p^r.\]
Write $f(x)$ as
\begin{equation}\label{eq:fexp0}
\begin{split}
f(x)&=\sum_{k=0}^{p-1}f_1(x)\cdot x^{p^{r-1}}+\dots+f_{p-1}(x)\cdot x^{(p-1)p^{r-1}},\\
&=\sum_{k=0}^{p-1}\sum_{\ell=0}^{p-1} f_{k,\ell}(x)\cdot x^{\ell p^{r-2}+k p^{r-1}},
\end{split}
\end{equation}
where $\deg(f_{k, \ell}) < p^{r-2}$ for $0\leq k, \ell \leq p-1$. For the convenience of writing, we sometimes denote $ x^p -\alpha x$ by $T$ and denote $\alpha^{p^m}$ by $\alpha_{m}$ for any integer $m\geq 0$. Since the characteristic of $\F_q$ is $p$, we have
\[x^{p^{r-1}}=(x^p-\alpha x)^{p^{r-2}}+\alpha_{r-2}x^{p^{r-2}}=T^{{p^{r-2}}}+\alpha_{r-2}x^{p^{r-2}}.\]
Thus, for each $k$ in $[0, p-1]$, we have
\begin{equation}\label{eq:T}
x^{kp^{r-1}}=(T^{p^{r-2}}+\alpha_{r-2}x^{p^{r-2}})^k=\sum_{j=0}^{k}\binom{k}{j}\cdot\alpha_{r-2}^{k-j}\cdot x^{(k-j)p^{r-2}} T^{jp^{r-2}}.
\end{equation}
Now, by substituting each $x^{kp^{r-1}}$ in equation \eqref{eq:fexp0} with the formula on the right side of equation \eqref{eq:fexp0}, we then have
\begin{equation}\label{eq:mid}
\begin{split}
f(x) &= \sum_{k, \ell}\sum_{j=0}^k f_{k,\ell}(x)\cdot \binom{k}{j}\cdot\alpha_{r-2}^{k-j}\cdot x^{(\ell+k-j)p^{r-2}} T^{jp^{r-2}}\\
&=\sum_{j=0}^{p-1} \sum_{k\geq j}^{p-1}\sum_{\ell=0}^{p-1}  f_{k,\ell}(x)\cdot \binom{k}{j}\cdot\alpha_{r-2}^{k-j}\cdot x^{(\ell+k-j)p^{r-2}} T^{jp^{r-2}}.
\end{split}
\end{equation}
If $\ell+k-j\geq p$, then
\[x^{(\ell+k-j)p^{r-2}}=x^{(\ell+k-j-p)p^{r-2}}(T^{p^{r-2}}+\alpha_{r-2}).\]
We split the sum over $\ell$ in equation \eqref{eq:mid} into two parts: a sum over $\ell$ satisfying $\ell\geq p-k+j$ and a sum over $\ell$ satisfying $\ell< p-k+j$. Then, by the above equation, we have
\begin{equation}\label{eq:final}
\begin{split}
f(x) &=\sum_{j=0}^{p-1}\left(\sum_{k\geq j}^{p-1}\sum_{\ell+k-j < p}  f_{k,\ell}(x)\cdot \binom{k}{j}\cdot\alpha_{r-2}^{k-j}\cdot x^{(\ell+k-j)p^{r-2}} \right)T^{jp^{r-2}} \\
&+\sum_{j=0}^{p-1} \left(\sum_{k\geq j}^{p-1}\sum_{\ell+k-j\geq p} f_{k,\ell}(x)\cdot \binom{k}{j}\cdot\alpha_{r-2}^{k-j+1}\cdot x^{(\ell+k-j-p+1)p^{r-2}}\right)T^{jp^{r-2}} \\
&+\sum_{j=1}^{p-1}\left( \sum_{k\geq j-1}^{p-1}\sum_{\ell+k-j+1\geq p}  f_{k,\ell}(x)\cdot \binom{k}{j-1}\cdot\alpha_{r-2}^{k-j+1}\cdot x^{(\ell+k-j-p+1)p^{r-2}} \right)T^{jp^{r-2}}.
\end{split}
\end{equation}
Note that the three coefficients of $T^{jp^{r-2}}$ in the above equation all are polynomials of degree less than $p^{r-1}$. Thus their sum, denoted by $g_j(x)$, is also a polynomial of degree less than $p^{r-1}$. To explicitly compute the coefficients polynomial $g_j(x)$ of $T^{jp^{r-2}}$, we need to add some polynomial terms $f_{k, \ell}x^{sp^{r-2}}$ and $f_{k',\ell'}x^{s'p^{r-2}}$ for $0\leq s, s'\leq p-1$ in equation \eqref{eq:final}. Since $\deg(f_{\ell, k}), \deg(f_{\ell', k'})<p^{r-2}$, the additions are only needed when $s=s'$. In this case, $f_{k, \ell}x^{sp^{r-2}}+f_{k',\ell'}x^{sp^{r-2}}$ can be computed in at most $p^{r-2}$ additions in $\F_q$. By counting, there are at most $p^2$ polynomials in the form $f_{k,\ell}x^{sp^{r-2}}$ for $0<\ell, k\leq p-1$. Thus we need at most $p\cdot p^{r}$ additions in $\F_q$ to compute the coefficient $g_j(x)$ of $T^{jp^{r-2}}$. Next, we count the total multiplications in $\F_q$ to compute $g_j(x)$. Since there are most $p^2$ polynomials $f_{k,\ell}$ needed to be multiplied by a scalar $\binom{k}{j}\alpha_{r-2}^{k-j} $ and each scalar multiplication needs at most $p^{r-2}$ multiplications in $\F_q$. Thus the total multiplications are also $p^{r+1}$. After the combinations, we get
\[f(x)= \sum_{j=0}^{p-1} g_j(x)T^{jp^{r-2}}, \ \deg(g_j)< p^{r-1}\ \text{for}\ 0\leq j\leq p-1.\]
Thus, the $(x^p-\alpha x)$-adic expansion of $f(x)$ of degree less than $p^r$ is reduced to $p$ $(x^p-\alpha x)$-adic expansions of polynomials of degree $< p^{r-1}$. We can continue this procedure recursively to $g_0, g_1, \dots, g_{p-1}$. Then in at most $r$ recursions, all the polynomials will have degrees less than $p$ and we obtain the $(x^p-\alpha x)$-expansion of $f(x)$. Let $A_p(n)$ and $M_p(n)$ be the complexity of the $(x^p-\alpha x)$-adic expansion of a polynomial in $\F_q[x]_{\leq n}$. Through the above analysis, we have
 \[A_p(n)=p\cdot A_p(n/p)+O(n),\ M_p(n)=p\cdot M_p(n/p)+O(n).\]
The recursive formula finally leads to $A_p(n)=O(n\cdot\log n)$ and $M_p(n)=O(n\cdot\log n)$.
\end{proof}

\section{An example of G-FFT with $q+1=128$}
\begin{table}[htbp]
  \centering
  \caption{The coefficients of $f(x)$}
    \begin{tabular}{|r|rrrrrrrrrrrrrrrrr|}
\hline
\hline
$e$  &     0     & 1     & 2     & 3     & 4     & 5     & 6     & 7     & 8     & 9     & 10    & 11    & 12    & 13    & 14    & 15    & 16 \\
$a_{\bbe}$ &      15    & 4     & 37    & 109   & 3     & 87    & 116   & 18    & 10    & 90    & 73    & 51    & 92    & 66    & 121   & 86    & 70 \\ \hline
 $e$  &    17    & 18    & 19    & 20    & 21    & 22    & 23    & 24    & 25    & 26    & 27    & 28    & 29    & 30    & 31    & 32    & 33 \\
 $a_{\bbe}$ &     13    & 21    & 95    & 29    & 122   & 78    & 122   & 78    & 41    & 26    & 49    & 44    & 66    & 19    & 66    & 40    & 121 \\  \hline
$e$  &       34    & 35    & 36    & 37    & 38    & 39    & 40    & 41    & 42    & 43    & 44    & 45    & 46    & 47    & 48    & 49    & 50 \\
   $a_{\bbe}$ &   81    & 3     & 116   & 4     & 50    & 40    & 121   & 85    & 25    & 66    & 38    & 55    & 42    & 98    & 37    & 116   & 15 \\ \hline
$e$  &     51    & 52    & 53    & 54    & 55    & 56    & 57    & 58    & 59    & 60    & 61    & 62    & 63    & 64    & 65    & 66    & 67 \\
  $a_{\bbe}$ &    49    & 33    & 100   & 86    & 120   & 104   & 61    & 114   & 0     & 10    & 17    & 68    & 91    & 81    & 98    & 124   & 44 \\ \hline
 $e$  &    68    & 69    & 70    & 71    & 72    & 73    & 74    & 75    & 76    & 77    & 78    & 79    & 80    & 81    & 82    & 83    & 84 \\
   $a_{\bbe}$ &    5     & 23    & 119   & 115   & 25    & 73    & 10    & 113   & 17    & 91    & 11    & 86    & 118   & 8     & 31    & 63    & 32 \\ \hline
$e$  &      85    & 86    & 87    & 88    & 89    & 90    & 91    & 92    & 93    & 94    & 95    & 96    & 97    & 98    & 99    & 100   & 101 \\
  $a_{\bbe}$ &     21    & 62    & 77    & 51    & 90    & 53    & 89    & 48    & 97    & 11    & 15    & 77    & 8     & 64    & 63    & 7     & 62 \\ \hline
 $e$  &     102   & 103   & 104   & 105   & 106   & 107   & 108   & 109   & 110   & 111   & 112   & 113   & 114   & 115   & 116   & 117   & 118 \\ 
  $a_{\bbe}$ &     55    & 92    & 116   & 116   & 118   & 53    & 80    & 39    & 47    & 84    & 53    & 100   & 4     & 97    & 40    & 106   & 108 \\ \hline
 $e$  &     119   & 120   & 121   & 122   & 123   & 124   & 125   & 126   & 127   &       &       &       &       &       &       &       &  \\
  $a_{\bbe}$ &     39    & 107   & 25    & 67    & 51    & 87    & 90    & 111   & 93    &       &       &       &       &       &       &       &  \\ \hline
    \end{tabular}%
  \label{tab:addlabel}%
\end{table}%

\begin{table}[htbp]
  \centering
  \caption{The MPEs of $f(x)$ and $\tilde{f}(x)$}
    \begin{tabular}{|r|rr|r|rr|r|rr|r|rr|r|rr|}
\hline 
\hline
$\alpha$ &  $f(\alpha)$ & $\tilde{f}(\alpha)$ & $\alpha$ & ${f}(\alpha)$ & $\tilde{f}(\alpha)$  & $\alpha$ & ${f}(\alpha)$ & $\tilde{f}(\alpha)$ & $\alpha$ & ${f}(\alpha)$ & $\tilde{f}(\alpha)$ & $\alpha$ & ${f}(\alpha)$ & $\tilde{f}(\alpha)$\\ \hline
{$\infty$} & 0     & 0     & 81    & 74    & 71    & 8     & 55    & 114   & 59    & 85    & 121   & 6     & 54    & 32 \\
    106   & 123   & 89    & 107   & 126   & 91    & 92    & 75    & 101   & 66    & 11    & 112   & 11    & 3     & 62 \\
    101   & 94    & 2     & 24    & 91    & 33    & 46    & 108   & 57    & 5     & 92    & 104   & 38    & 61    & 123 \\
    111   & 123   & 108   & 49    & 3     & 63    & 117   & 99    & 31    & 22    & 85    & 35    & 41    & 10    & 79 \\
    21    & 42    & 64    & 68    & 105   & 95    & 9     & 96    & 83    & 20    & 68    & 5     & 19    & 119   & 34 \\
    54    & 102   & 106   & 110   & 42    & 107   & 84    & 22    & 18    & 93    & 74    & 14    & 26    & 117   & 53 \\
    31    & 9     & 99    & 1     & 40    & 102   & 37    & 68    & 9     & 25    & 98    & 111   & 27    & 25    & 48 \\
    64    & 98    & 107   & 76    & 93    & 5     & 99    & 77    & 6     & 53    & 96    & 62    & 124   & 68    & 43 \\
    34    & 114   & 12    & 48    & 109   & 91    & 12    & 31    & 68    & 75    & 124   & 119   & 10    & 40    & 22 \\
    94    & 38    & 103   & 113   & 10    & 75    & 86    & 36    & 50    & 115   & 41    & 5     & 97    & 40    & 104 \\
    89    & 89    & 93    & 13    & 16    & 100   & 72    & 30    & 124   & 23    & 68    & 48    & 32    & 116   & 22 \\
    100   & 120   & 30    & 109   & 21    & 100   & 103   & 68    & 94    & 91    & 48    & 71    & 60    & 72    & 66 \\
    51    & 76    & 8     & 73    & 113   & 84    & 0     & 86    & 61    & 40    & 83    & 72    & 63    & 20    & 68 \\
    118   & 37    & 117   & 126   & 48    & 109   & 2     & 79    & 89    & 116   & 96    & 30    & 80    & 3     & 31 \\
    112   & 0     & 0     & 14    & 29    & 116   & 70    & 23    & 71    & 30    & 23    & 17    & 65    & 54    & 60 \\
    123   & 24    & 95    & 69    & 36    & 71    & 82    & 120   & 118   & 108   & 69    & 33    & 119   & 40    & 11 \\
    4     & 0     & 0     & 57    & 77    & 83    & 43    & 89    & 11    & 33    & 75    & 81    & 45    & 111   & 107 \\
    105   & 111   & 59    & 78    & 106   & 106   & 56    & 105   & 52    & 122   & 73    & 33    & 96    & 48    & 15 \\
    35    & 109   & 6     & 39    & 35    & 22    & 98    & 15    & 65    & 3     & 104   & 43    & 62    & 126   & 74 \\
    67    & 51    & 57    & 95    & 65    & 114   & 125   & 8     & 126   & 15    & 110   & 69    & 121   & 58    & 17 \\
    17    & 62    & 44    & 77    & 73    & 22    & 44    & 58    & 1     & 83    & 46    & 47    & 52    & 93    & 9 \\
    102   & 105   & 77    & 120   & 53    & 90    & 47    & 10    & 16    & 85    & 16    & 97    & 90    & 36    & 105 \\
    36    & 69    & 52    & 7     & 101   & 101   & 74    & 110   & 72    & 29    & 76    & 109   & 55    & 77    & 79 \\
    61    & 40    & 48    & 28    & 77    & 83    & 79    & 69    & 55    & 42    & 122   & 85    & 104   & 34    & 77 \\
    18    & 89    & 92    & 16    & 41    & 35    & 58    & 77    & 31    & 87    & 6     & 31    &       &       &  \\
    50    & 123   & 86    & 71    & 84    & 82    & 88    & 4     & 10    & 114   & 43    & 17    &       &       &  \\ \hline
    \end{tabular}%
  \label{tab:addlabel}%
\end{table}%

\end{document}